\documentclass[a4paper]{article}

\usepackage[l2tabu, orthodox]{nag}
\usepackage[utf8]{inputenc}
\usepackage[pdfusetitle]{hyperref}
\usepackage{microtype}

\usepackage{mathtools}

\mathtoolsset{
  showonlyrefs=true % or false in draft mode
}%use \equation{} everywhere; only the referenced ones are numbered

\usepackage{amsmath,amssymb,amsthm}
\usepackage{thm-restate}

\newtheorem{lemma}{Lemma}
\newtheorem{theorem}{Theorem}

\newtheorem{corollary}{Corollary}
\theoremstyle{definition}

\theoremstyle{remark}

\usepackage{kbordermatrix}

\renewcommand{\kbldelim}{(}% Left delimiter
\renewcommand{\kbrdelim}{)}% Right delimiter

\usepackage{fullpage}

\usepackage{enumerate} %Gives ability, e.g., \begin{enumerate}[Case 1.]

\usepackage[noend]{algpseudocode}
\usepackage{algorithm}

\usepackage{xspace}

\newcommand{\N}{\ensuremath{\mathbb N}\xspace}
\newcommand{\R}{\ensuremath{\mathbb R}\xspace}
\newcommand{\F}{\ensuremath{\mathbb F}\xspace}
\newcommand{\Q}{\ensuremath{\mathbb Q}\xspace}

\newcommand{\cI}{\ensuremath{\mathcal I}\xspace}
\newcommand{\cB}{\ensuremath{\mathcal B}\xspace}

\newcommand{\cF}{\ensuremath{\mathcal F}\xspace}

\newcommand{\cS}{\ensuremath{\mathcal S}\xspace}

\newcommand{\bD}{\ensuremath{\mathbf D}\xspace}

\newcommand{\Oh}{O}

\DeclareMathOperator{\rank}{rank}

\usepackage[usenames,dvipsnames]{color}

\usepackage{todonotes}

\DeclareMathOperator{\Pf}{Pf}

\title{Faster algorithms on linear delta-matroids}
\author{Tomohiro Koana\thanks{Algorithmics and Computational Complexity, Technische Universität Berlin, \texttt{tomohiro.koana@tu-berlin.de}}
 \and
 Magnus Wahlström\thanks{Dept. of Computer Science, Royal Holloway,
   University of London, \texttt{Magnus.Wahlstrom@rhul.ac.uk}}}

\begin{document}
\maketitle

\begin{abstract}
  We show new algorithms and constructions over linear delta-matroids.
  We observe an alternative representation for linear delta-matroids,
  as a \emph{contraction representation} over a skew-symmetric matrix.
  This is equivalent to the more standard ``twist representation''
  up to $O(n^\omega)$-time transformations (where $n$ is the dimension
  of the delta-matroid and $\omega < 2.372$ the matrix multiplication exponent), 
  but is much more convenient for algorithmic tasks.
  For instance, the problem of finding a max-weight feasible set now
  reduces directly to the problem of finding a max-weight basis
  in a linear matroid.
  Supported by this representation, we provide new algorithms and
  constructions over linear delta-matroids. We show that the \emph{union}
  and \emph{delta-sum} of linear delta-matroids define linear
  delta-matroids, and a representation for the resulting
  delta-matroid can be constructed in randomized time $O(n^\omega)$
  (or more precisely, in $O(n^{\omega})$ field operations,
  over a field of at least $\Omega(n \cdot (1/\varepsilon))$
  elements, where $\varepsilon > 0$ is an error parameter). 
  Previously, it was only known that these operations define delta-matroids.
  We also note that every \emph{projected} linear delta-matroid
  can be represented as an \emph{elementary} projection.
  This implies that several optimization problems over (projected) linear delta-matroids,
  including the coverage, delta-coverage, and parity problems,
  reduce (in their decision versions) to a single $O(n^{\omega})$-time
  matrix rank computation.
  Using the methods of Harvey,
  previously used by Cheung, Lao and Leung for linear matroid parity,
  we furthermore show how to solve the search versions in the same time.
  This improves on the $O(n^4)$-time augmenting path algorithm 
  of Geelen, Iwata and Murota.   
  Finally, we consider the maximum-cardinality delta-matroid
  intersection problem (or equivalently, the maximum-cardinality
  delta-matroid matching problem).
  Using Storjohann's algorithms for symbolic determinants,
  we show that such a solution can be found in $O(n^{\omega+1})$ time.
  This is the first polynomial-time algorithm for the problem,
  solving an open question of Kakimura and Takamatsu.
\end{abstract}

\section{Introduction}

Matroids are important unifying structures in many parts of computer
science and discrete mathematics, abstracting and generalizing
notions from linear vector spaces and graph theory;
see, e.g., Oxley~\cite{OxleyBook2} and Schrijver~\cite{SchrijverBook}.
Formally, a matroid is a collection of \emph{independent sets},
subject to particular axioms (see below).
A maximum independent set is a \emph{basis}. 
Among other things, matroids are a very useful source of algorithmic
meta-results, since there are many problems on matroids which admit 
efficient, general-purpose algorithms -- such as the greedy
algorithm for finding a max-weight basis, generalizing
algorithms for max-weight spanning trees; or 
\textsc{Matroid Intersection}, the problem of finding a maximum
common feasible set in two given matroids, which generalizes
bipartite matching. 

An important class of matroids are \emph{linear matroids},
where independence in the matroid is represented by the column space of
a matrix. Linear matroids enjoy many properties not shared by generic
matroids. For example, the famous \textsc{Matroid Parity} problem,
which generalizes the matching problem in general graphs,
is known to be intractable in the general case
but efficiently solvable over linear matroids~\cite{Lovasz80matroid}.
In addition, linear representations, as a compact representation
of combinatorial information, have seen many applications
in parameterized complexity for purposes of 
\emph{sparsification} and \emph{kernelization}~\cite{KratschW14TALG,KratschW20JACM},
and algebraic algorithms over linear matroids have proven a very
useful general-purpose tool in FPT algorithms~\cite{FominLPS16JACM,EibenKW24SODA}
(cf.~\cite{CyganFKLMPPS15PCbook,FominLSZkernelsbook}).

\emph{Delta-matroids} are a generalization of matroids,
where, informally, the notion of bases
is replaced by the notion of \emph{feasible sets},
which satisfy an exchange axiom similar to matroid bases
but need not all have the same cardinality.
Delta-matroids were introduced by Bouchet (although similar structures
were independently defined by others),
and have connections to multiple areas of computer science such as
structural and topological graph theory~\cite{Moffatt19deltamatroids}, 
constraint satisfaction problems~\cite{Feder01fanout,KazdaKR19},
%%combinatorial optimization? 
matching and path-packing problems and more; see below.
Like with matroids, there is also a notion of \emph{linear delta-matroids},
where the feasible sets are represented through a skew-symmetric matrix.
These generalize linear matroids, although this fact
(or indeed the fact that skew-symmetric matrices define
delta-matroids) is not elementary~\cite{GeelenIM03}.
Delta-matroids (linear or otherwise) are remarkably flexible structures,
in that there are many ways to modify or combine given delta-matroids
into new delta-matroids, including twisting (partial dualization),
contraction and deletion, existential projection,
and unions and delta-sums of delta-matroids (all described below).

Similarly to matroids, there is also a range of generic problems that
have been considered over delta-matroids, including
delta-matroid intersection, partition, and parity problems.
Unfortunately, due to the generality of delta-matroids,
these problems are all intractable in the general case, since they
generalize matroid parity.
%(Moreover, they are to some extent inter-reducible.)
However, they are tractable on linear delta-matroids, where
Geelen et al.~\cite{GeelenIM03} gave an algorithm (and a corresponding
min-max theorem) with a running time of $O(n^4)$, improvable to
$O(n^{\omega+1})$ using fast matrix multiplication.
However, other variants remain open. 
Kakimura and Takamatsu~\cite{KakimuraT14sidma} considered the maximum
cardinality version of delta-matroid parity (as opposed to the result
of Geelen et al.~\cite{GeelenIM03}, which is more of a \emph{feasibility}
or \emph{minimum error} version). They gave a solution for a
restricted class of linear and projected linear delta-matroid, 
but left the general case open. Furthermore, the natural weighted
optimization variants of the above appear completely open.

In this paper, we show new constructions of linear delta-matroids
and new and faster algorithms for the aforementioned problems on
linear and projected linear delta-matroids.
In particular, we show a new representation variant for linear
delta-matroids -- dubbed \emph{contraction representation},
as opposed to the standard \emph{twist representation} -- 
which appears more amenable to efficient algorithms. 
Using this representation, we show for the first time
that unions and delta-sums of linear delta-matroids
(represented over a common field $\F$) define \emph{linear}
delta-matroids, and that a representation can be constructed
in randomized polynomial time.
We also show new algorithmic results,
including solving the search version of
\textsc{Linear Delta-Matroid Parity} (\textsc{Linear DM Parity} for short) in $O(n^\omega)$ field operations\footnote{Throughout, we give our running times as field operations. If the field has size $n^{O(1)}$, as in most applications, then this is just an $\tilde O(1)$ overhead, but if a  representation is provided over an enormous field, the overhead is naturally larger.}
and giving the first polynomial-time algorithm for
the maximum cardinality version of the problem in $O(n^{\omega+1})$ field operations,
thereby settling an open question from Kakimura and Takamatsu~\cite{KakimuraT14sidma}.

\subsection{Introduction to delta-matroids}

Before we describe our results in detail, let us review some
background on delta-matroids. For more material, we refer to the
survey by Moffatt~\cite{Moffatt19deltamatroids}.

Like matroids, delta-matroids are formally defined as set systems
satisfying particular axioms.
Formally, a delta-matroid is a pair $D=(V,\cF)$ where $V$ is a ground set
and $\cF \subseteq 2^V$ a non-empty collection of subsets of $V$, referred to as
\emph{feasible sets} in $D$, subject to the following \emph{symmetric exchange axiom}:
\begin{equation}
  \label{eq:sea}
  \text{For all } F_1, F_2 \in \cF,\, x \in F_1 \Delta F_2
  \text{ there exists } y \in F_1 \Delta F_2 \text{ such that } F_1 \Delta \{x,y\} \in \cF,
\end{equation}
where $\Delta$ denotes symmetric difference.

It should be enlightening to compare this to the definition of matroids.
Formally, a matroid is most commonly defined as a collection of \emph{independent sets};
i.e., a matroid is defined as a pair $M=(V,\cI)$ where $V$ is the ground set and
$\cI \subseteq 2^V$ is a collection of sets, referred to as
independent sets in $M$, subject to (1) $\emptyset \in \cI$;
(2) if $B \in \cI$ and $A \subset B$ then $A \in \cI$;
and (3) if $A, B \in \cI$ with $|A|<|B|$ then there
exists an element $x \in B \setminus A$ such that $A+x \in \cI$.
The second condition encodes that $(V,\cI)$ is an \emph{independence system}.
However, matroids can also be equivalently defined from
just the collection of \emph{maximal} independent sets,
known as \emph{bases}. Under this definition, a matroid
is a pair $M=(V,\cB)$ where $\cB \subseteq 2^V$
is a non-empty collection of bases, subject to the
\emph{basis exchange property}:
\begin{equation}
  \label{eq:bep}
  \text{For all } A, B \in \cB,\, x \in A \setminus B
  \text{ there exists } y \in B \setminus A \text{ such that } A \Delta \{x,y\} \in \cB.
\end{equation}
In particular, all bases of a matroid have the same cardinality.
Thus, delta-matroids can be seen as the relaxation of matroids
where the feasible sets need not all have the same cardinality. 
In fact, a delta-matroid where all feasible sets have the same
cardinality is precisely a matroid. (Similarly, the set of independent sets
of a matroid forms the feasible sets of a delta-matroid, and furthermore
a delta-matroid which is an independence system is precisely a matroid in this sense.
But the formulation from the set of bases is more convenient.)
As a further illustration, consider the case of graph matchings. 
Let $G=(V,E)$ be a graph. The \emph{matching matroid} of $G$ is a
matroid over ground set $V$ where a set $B \subseteq V$ is a basis if
and only if it is the set of endpoints of a maximum matching of $G$.
Correspondingly,  the independent sets $S \subseteq V$ of the matching
matroid are vertex sets that can be covered by a matching.
On the other hand, the \emph{matching delta-matroid} over $G$
is the delta-matroid where a set $S \subseteq V$ is feasible
if and only if $G[S]$ has a perfect matching.
Thus, the maximal feasible sets of the matching delta-matroid
form the bases of the matching matroid, but clearly,
the matching delta-matroid captures more of the structure
of $G$ than the matching matroid does. 

\paragraph*{Linear delta-matroids.}
As with matroids, an important class of delta-matroids are
\emph{linear delta-matroids}. A matrix $A$ is \emph{skew-symmetric}
if $A^T=-A$. Let $A$ be a skew-symmetric matrix
with rows and columns indexed by a set $V$. Then $A$ defines a
delta-matroid $\bD(A)=(V,\cF)$ where for $S \subseteq V$ we have
$S \in \cF$ if and only if $A[S]$ is non-singular. We refer to
$\bD(A)$ as a \emph{directly represented} linear delta-matroid.
More generally, the \emph{twist} of a delta-matroid $D=(V,\cF)$ by a set $S \subseteq V$,
denoted $D \Delta S$, is the delta-matroid with feasible sets
\[
  \cF \Delta S := \{F \Delta S \mid F \in \cF\}.
\]
It is easy to check that $D \Delta S$ is a delta-matroid.
The twisting operation is also known as \emph{partial dualization}, since
the twist $D^* := D \Delta V$ corresponds to the dualization $M^*$ of
a matroid $M$. A general representation of a linear delta-matroid is given as
$D=\bD(A) \Delta S$ for some skew-symmetric matrix $A$ and twisting
set $S$. A delta-matroid $D$ is \emph{even} if all feasible sets have
the same cardinality; all linear delta-matroids are even. 
In addition, we consider \emph{projected linear delta-matroids},
which is a delta-matroid $D=(V,\cF)$ defined via existential projection
over a set $X$ from a larger linear delta-matroid $D'=\bD(A) \Delta S$
over ground set $V \cup X$. We denote this $D=D'|X$, where $D$
has the feasible set
\[
  \cF = \{F \setminus X \mid F \in \cF(D')\}.
\]
As a canonical example, the matching delta-matroid of a graph $G$
is directly represented by the Tutte matrix over $G$.
The set of bases of a linear matroid forms a linear delta-matroid,
and the independent sets form a projected linear delta-matroid,
under natural representations; see the end of Section~\ref{sec:contraction}.

Underpinning algorithms on linear delta-matroids are a number of
fundamental operations on skew-symmetric matrices. 
For a skew-symmetric matrix $A$ indexed by $V$
and a set $S \subseteq V$ such that $A[S]$ is non-singular,
there is a \emph{pivoting} operation that constructs a new
skew-symmetric matrix $A'=A*S$ such that for any $U \subseteq V$,
$A'[U]$ is non-singular if and only if $A[S \Delta U]$ is.
Via this operation, linear delta-matroids are closed under the
\emph{contraction} operation $D/T$ as well as \emph{deletion} $D \setminus T$.
Another fundamental property of skew-symmetric matrices is the \emph{Pfaffian},
defined as follows. Let $A$ be a skew-symmetric matrix with rows and
columns indexed by $V$. The \emph{support graph} of $A$ is the graph $G=(V,E)$
where $uv \in E$ if and only if $A[u,v] \neq 0$. Then the Pfaffian of $A$ is
defined as
\begin{align*}
  \Pf A = \sum_{M} \sigma(M) \prod_{e \in M} A[u, v],
\end{align*}
where $M$ ranges over all perfect matchings in $G$ and
$\sigma(M) \in \{1,-1\}$ is a sign term. It holds that
$\det A = (\Pf A)^2$, thus $A$ is non-singular if and only if $\Pf A \neq 0$.
Via this connection to matchings, the Pfaffian forms a link between
the combinatorial and algebraic aspects of linear delta-matroids,
in a way that is often exploited in this paper. 
The Pfaffian also enjoys some useful algebraic properties,
such as the Pfaffian sum formula and the Ishikawa-Wakayama formula,
with clear combinatorial interpretations. See Section~\ref{sec:prelim}
for details.

\subsection{Our results}

We show a range of results regarding the representation and
construction of linear delta-matroids, and new and faster algorithms
for computational problems over them.
We discuss these in turn. 

\subsubsection{Representations and constructions}

Our first result, which supports the others, is the introduction of a
new representation for linear delta-matroids. Recall that a linear
delta-matroid $D=(V,\cF)$ is represented as $D=\bD(A) \Delta S$
for a skew-symmetric matrix $A$ with rows and columns indexed by $V$.
We refer to this as a \emph{twist representation}. Although this
representation is intimately connected to the structure of delta-matroids,
it is less convenient for algorithmic purposes.
For this, we introduce the \emph{contraction representation},
representing a linear delta-matroid $D=(V,\cF)$
as $D=\bD(A)/T$ for a skew-symmetric matrix $A$ with rows and columns
indexed by $V \cup T$. Thus, a set $F \subseteq V$ is feasible in $D$
if and only if $A[F \cup T]$ is non-singular. 

We show that the representations are equivalent, and given a
representation in one form, we can efficiently and deterministically construct one in the other;
see Section~\ref{sec:contraction}.
Thus contraction representations do not change the class of representable
delta-matroids; however, we find that the contraction representation maps
much more directly into the algorithmic methods of linear algebra.

Next, we consider two methods of composing linear delta-matroids.
Let $D_1=(V,\cF_1)$ and $D_2=(V,\cF_2)$ be given delta-matroids
(padding the ground sets with dummy elements if necessary so that they
are defined over the same ground set $V$).
The \emph{union} $D_1 \cup D_2$ is the delta-matroid $D=(V,\cF)$
where
\[
  \cF = \{F_1 \cup F_2 \mid F_1 \in \cF_1, F_2 \in \cF_2, F_1 \cap F_2 = \emptyset\},
\]
i.e., the feasible sets in $D$ are the disjoint unions of feasible
sets in $D_1$ and $D_2$. Additionally, the \emph{delta-sum} $D_1 \Delta D_2$ 
is defined as the delta-matroid $D=(V,\cF)$ with feasible sets
\[
  \cF=\{F_1 \Delta F_2 \mid F_1 \in \cF_1, F_2 \in \cF_2\},
\]
where $\Delta$ denotes symmetric difference.
Bouchet~\cite{Bouchet89dam} and Bouchet and Schwärzler~\cite{BouchetS98deltasum}
showed that $D_1 \cup D_2$ and $D_1 \Delta D_2$ are delta-matroids.
Using properties of Pfaffians and the contraction representation, we show 
that furthermore, the union and delta-sum of linear delta-matroids are
linear delta-matroids.

The construction is randomized, and takes an error parameter
$\varepsilon > 0$ which controls the size of the field that the output
delta-matroid is represented over. For this purpose, we say that an
algorithm constructs an \emph{$\varepsilon$-approximate representation}
of a delta-matroid $D=(V,\cF)$ if it constructs a representation of a
delta-matroid $D'=(V,\cF')$ where $\cF' \subseteq \cF$ and for every
$F \in \cF$ the probability that  $F \in \cF'$ is at least $1-\varepsilon$.
%This randomness appears
%unavoidable with current tools, since the matching delta-matroid for a graph $G=(V,E)$
%(i.e., the Tutte matrix of $G$)
%can be constructed through a union over $|E|$ delta-matroids with
%trivial representations.
%
%the Tutte matrix of a graph $G$
%(representing the matching delta-matroid of $G$) is the delta-matroid
%union over elementary delta-matroids with feasible sets
%$\{\emptyset, \{u,v\}\}$, $uv \in E(G)$.\footnote{It is easy to see
%  that any representation of the matching delta-matroid of $G$ yields
%  an instance of the Tutte matrix of $G$.}
%
Setting $\varepsilon = O(1/2^n)$ where $n=|V|$ gives a representation
that with good probability is correct for all subsets. However, this
leads to a prohibitive field size, with significant 
overhead cost per field operation. Thus, for algorithmic
applications, a smaller value of $\varepsilon$ may be faster and
sufficient.

\begin{theorem} \label{ithm:union} \label{ithm:delta-sum}
  Let $D_1$ and $D_2$ be delta-matroids represented over a common field $\F$,
  and let $\varepsilon > 0$ be given. Let $\F'$ be an extension field of $\F$ 
  with at least $n \cdot \lceil 1/\varepsilon \rceil$ elements. 
  Then the delta-matroid union $D_1 \cup D_2$
  and delta-sum $D_1 \Delta D_2$ 
  are linear, and
  $\varepsilon$-approximate representations over $\F'$
  can be constructed in $O(n^2)$ respectively $O(n^\omega)$ field operations.
\end{theorem}

Additionally, Bouchet and Cunningham~\cite{BouchetC95} defined the
\emph{composition} of two delta-matroids $D_1=(V_1,\cF_1)$
and $D_2=(V_2,\cF_2)$ with partially overlapping ground sets
as the set system
$
  \cF=\{F_1 \Delta F_2 \mid F_1 \cap (V_1 \cap V_2)=F_2 \cap (V_1 \cap V_2)\}.
$
%This has seen some application in CSP context~\cite{KazdaKR19}.
Since it is equivalent to $(D_1 \Delta D_2) \setminus (V_1 \cap V_2)$,
it is covered by the above results. 

We remark that the more immediate interpretation of $D_1 \cup D_2$ as containing
all sets $F_1 \cup F_2$ where $F_1 \in \cF(D_1)$ and $F_2 \in \cF(D_2)$,
which is closer to how matroid union is defined, does not define a
delta-matroid.
%Indeed, consider $\cF_1=\{\emptyset, \{u,v\}\}$
%and $\cF_2=\{\emptyset,\{v,w\}\}$. Then their union would contain
%$F=\emptyset$ and $F'=\{u,v,w\}$, but there is no $y \in F \Delta F'$
%such that $F' \Delta \{v,y\}$ is the union of two feasible sets in
%$\cF_1$ and $\cF_2$.
(On the other hand, the delta-matroid union of
the independent sets of two matroids $M_1$, $M_2$ produces
the independent sets of the matroid union $M_1 \lor M_2$.)

All the above results easily extend to projected linear delta-matroids.

\subsubsection{Algorithms}

As a warm-up, we first consider the problem of finding a max-weight
feasible set in a given delta-matroid $D=(V,\cF)$ with element weights $w \colon V \to \R$.
Note that the weights may be negative, and since
not all feasible sets have the same cardinality, unlike in matroids,
they cannot simply be raised to be non-negative.
Bouchet~\cite{Bouchet95,Bouchet87DMone} showed that there is a variant
of greedy algorithm that solves this problem using only separation
oracle calls.
%(in fact, delta-matroids are precisely the set system where this greedy algorithm succeeds).
However, this requires $O(n)$ separation oracle calls. 
We show that, using the contraction representation,
the max-weight feasible set problem in a linear delta-matroid
reduces to
%the max-weight basis problem in a linear \emph{matroid}
%(i.e.,
finding a max-weight column basis of an $O(n) \times O(n)$ matrix, 
which can be done significantly faster. 

\begin{theorem} \label{ithm:weighed-feasible}
  Let $D=(V,\cF)$ be a linear or projected linear delta-matroid.
  In $O(n^\omega)$ field operations, we can find a max-weight
  feasible set in $D$. 
\end{theorem}

%\begin{itemize}
%  \item The ``contraction representation'' of linear delta-matroids.
%  \item Linear representation of delta-sum of two linear delta-matroids
%  \item A ``compact'' representation of projected linear delta-matroids
%  \item $\Oh(n^{\omega})$-time algorithm to find a max-weight feasible set
%  \item $\Oh(n^{\omega})$-time algorithm to find feasible sets $F_1 \in \cF_1$ and $F_2 \in \cF_2$ maximizing $F_1 \Delta F_2$
%  \item $\Oh(n^{\omega})$-time algorithm to find a common feasible set and $O(n^3)$-time to find a maximum-cardinality common feasible set
%\end{itemize}

For more intricate questions,
the literature contains a range of problems over delta-matroids.
Let $D_1=(V,\cF_1)$ and $D_2=(V,\cF_2)$ be given delta-matroids. The
following are some key problems~\cite{Bouchet95,GeelenIM03}.
%We use DM to abbreviate delta-matroid.
\begin{itemize}
\item \textsc{DM Covering}:  Given $D_1$ and $D_2$, find $F_1 \in \cF_1$, $F_2 \in \cF_2$
  with $F_1 \cap F_2=\emptyset$ to maximize $|F_1 \cup F_2|$
\item \textsc{DM Delta-Covering}: Given $D_1$ and $D_2$, find $F_1 \in \cF_1$, $F_2 \in \cF_2$
  to maximize $|F_1 \Delta F_2|$
\item \textsc{DM Intersection}: Given $D_1$ and $D_2$, find a common feasible set $F \in \cF_1
  \cap \cF_2$
\item \textsc{DM Partition}: Given $D_1$ and $D_2$, find a partition $V=P \cup Q$
  such that $P \in \cF_1$ and $Q \in \cF_2$
\item \textsc{DM Parity}: Given a delta-matroid $D=(V,\cF)$ and a partition $\Pi$ of $V$ into pairs,
  is there a feasible set in $D$ which is the union of pairs? More
  generally, find a feasible set $F \in \cF$ to minimize
  the number of broken pairs, i.e., the number of pairs $p \in \Pi$ with
  $|p \cap F|=1$.
\end{itemize}
The variant of \textsc{DM Covering} where the disjointness constraint
is dropped reduces to \textsc{Matroid Union}, since the maximal
feasible sets of a delta-matroid form a matroid, hence is of less
interest for delta-matroids.

Using the methods of the previous subsection,
the decision versions of the above for linear and projected linear
delta-matroids all reduce to computing the  rank of an $O(n) \times O(n)$ skew-symmetric matrix.
Indeed, consider \textsc{DM Delta-Covering}. Assume that $D_1$ and $D_2$ 
are given in some linear representation and let $D=\bD(A)/T=D_1 \Delta D_2$. 
Then a set $F \subseteq V$ is a solution $F=F_1 \Delta F_2$ to the
delta-covering problem if and only if $F \cup T$ is a basis of $A$. 
\textsc{DM Covering} and \textsc{DM Partition} work similarly,
and \textsc{DM Intersection} is asking whether $\emptyset$ is feasible in $D_1 \Delta D_2$.

For \textsc{DM Parity}, assume w.l.o.g.\ that $V=\{u_1,v_1,\ldots,u_n,v_n\}$
and $\Pi=\{\{u_i,v_i\} \mid i \in [n]\}$. Let $D_\Pi$ be the matching
delta-matroid for the graph with edge set $\Pi$, so that the feasible
sets are precisely sets $F \subseteq V$ where
$u_i \in F$ if and only if $v_i \in F$ for $i \in [n]$. 
Then \textsc{DM Parity} reduces to finding the rank
of $D \Delta D_\Pi$ (or indeed $D \cup D_\Pi$). 
Thus the decision versions of all the above problems can be solved by a
randomized algorithm using $O(n^\omega)$ field operations
given linear representations of $D$ resp.\ $D_1$ and $D_2$.

Furthermore, the problems inter-reduce. We saw that \textsc{DM Parity}
reduces to \textsc{DM Delta-Covering}. In the other direction,
let $(D_1,D_2)$ be an instance of \textsc{DM Delta-Covering}.
%Clearly \textsc{DM Parity} is also
%the intersection problem between $D$ and $D_\Pi$,
%and \textsc{Partition} is equivalent to the intersection problem on
%$D_1$ and the dual $D_2^*$ of $D_2$.
%Finally, \textsc{Delta-Covering} reduces to \textsc{Parity}:
Let $D$ be the disjoint union of $D_1$ and $D_2^*$ on copies of the
ground set $V$ and let $\Pi$ be the pairing which links up the two
copies in $D$ of every $v \in V$. Then, for $F_1 \in \cF(D_1)$
and $F_2^* \in \cF(D_2^*)$, the pair %$p_v \in \Pi$
representing $v \in V$ is broken in $F_1 \cup F_2^*$
precisely if $v \notin F_1 \Delta F_2$,
where $F_2 = V \setminus F_2^*$; note that $F_2 \in \cF(D_2)$. 
For instances with no broken pairs, similar reductions apply for 
\textsc{DM Covering}, \textsc{DM Intersection} and \textsc{DM Partition}.
Since \textsc{DM Parity} directly generalizes \textsc{Matroid Parity},
all these problems are thus worst-case intractable.

For the search versions, the picture gets a little different.
In particular, the \textsc{DM Parity}, \textsc{DM Delta-Covering} and \textsc{DM Covering}
problems have two notions of witnesses. There is a \emph{weak}
witness, which for \textsc{DM Parity} is the set of broken pairs in some
optimal solution and for \textsc{DM Delta-Covering} and \textsc{DM Covering} simply the set $F$
for an optimal solution $F=F_1 \Delta F_2$ resp.~$F=F_1 \cup F_2$;
and a \emph{strong} witness, which for \textsc{DM Parity} is the set $F$
and for \textsc{DM Delta-Covering} and \textsc{DM Covering} the pair $(F_1,F_2)$. 
It can be verified that search problems for weak witnesses
reduce to each other as above, and the weak witness can easily be
recovered from the above-mentioned matrix representation of $D$.
However, this does not recover the strong witness other than via
self-reducibility. In turn, being able to compute strong witness for
either of these problems implies being able to compute either type of
witness for all five problems.
Hence the more interesting search problem is the latter.

Geelen et al.~\cite{GeelenIM03} showed an algorithm for
computing the strong witness for \textsc{Linear DM Parity}, 
which can be implemented in $O(n^{\omega+1})$ field operations.
Applying self-reducibility over the above-mentioned
representation would already reproduce the same running time. Using methods of
Harvey~\cite{Harvey09}, which also underpin the currently fastest
algorithm for linear matroid parity~\cite{CheungLL14}, we show the
following improved result.  
The condition of field size is for simplicity; given representations
over a common field $\F$ we can easily move to a large enough
extension field of $\F$.

\begin{restatable}{theorem}{targetdeltasum} \label{ithm:target-delta-sum}
  Let $D_1=(V,\cF_1)$ and $D_2=(V,\cF_2)$ be linear or projected
  linear delta-matroids over a common field with $\Omega(n^3)$ elements and let $S \subseteq V$
  be feasible in $D_1 \Delta D_2$.
  In $O(n^\omega)$ field operations, we can find with high probability feasible sets
  $F_1 \in \cF_1$ and $F_2 \in \cF_2$ such that $F_1 \Delta F_2 = S$. 
\end{restatable}

If $S$ is the weak witness for \textsc{Linear DM Delta-Covering},
we thus recover the strong (actual) witness in the same time.
For \textsc{Linear DM Covering}, we would apply this to the
subinstance induced by the weak witness $F$. 
The other problems follow as above.

\begin{corollary} \label{icor:search}
  The following search problems can be solved in $O(n^\omega)$ field
  operations with high probability, given (projected) linear delta-matroids on ground sets of $n$ elements represented over
  a common field with $\Omega(n^3)$ elements: \textsc{DM Covering}; \textsc{DM Delta-Covering};
  \textsc{DM Intersection}; \textsc{DM Partition}; and \textsc{DM Parity}.
\end{corollary}

Finally, we consider the \emph{weighted} versions of the above problems.
Again, the picture becomes slightly different. For \textsc{DM Partition},
there is no sensible weighted version. For \textsc{DM Covering} and \textsc{DM Delta-Covering},
the natural weight is the weight of the solution set $F$,
in which case the problem is solved in
strongly polynomial time of $O(n^\omega)$ field operations
using Theorem~\ref{ithm:weighed-feasible} and~\ref{ithm:target-delta-sum}.
For \textsc{DM Parity}, we may consider attaching weights to the weak or
the strong witness. In the former case, we end up with a problem where
you have to pay for breaking a pair and want to minimize the cost of a
feasible set; again, this is solved in $O(n^\omega)$ field operations
using Theorem~\ref{ithm:weighed-feasible} and~\ref{ithm:target-delta-sum}.
For the more interesting version, we assume that there exists a
delta-matroid parity set with no broken pairs, and we wish to maximize
the weight of such a set. This version, finally, is equivalent to the
weighted version of \textsc{DM Intersection}, which we
focus on, and directly generalizes \textsc{Weighted Matroid Parity}.
We use the matrix representation of the problem to construct 
a solution via algebraic methods.
%Define an auxiliary variable $z$ and for each element $v$ with weight $w(v)$,
%multiply its corresponding occurrences in the matrix formulation by $z^{w(v)}$. 
%We can then look for the maximum power $z^W$ of $z$ that has a
%non-zero coefficient in the Pfaffian. To avoid running time overhead
%due to polynomial interpolation, we compute this using
%Storjohann's symbolic determinant algorithm~\cite{Storjohann03} 

As a special case, even the unit weight version,
\textsc{Maximum DM Intersection},
is an interesting problem which up to now has had no polynomial-time
solution. Kakimura and Takamatsu~\cite{KakimuraT14sidma}
asked this as an open problem, and provided algorithms for some
special cases of it. We solve the general case.

\begin{restatable}{theorem}{cardinalityintersection} \label{ithm:cardinality-intersection}
  Let $D_1=(V,\cF_1)$ and $D_2=(V,\cF_2)$ be linear or projected
  linear delta-matroids over a common field with $\Omega(n^2)$ elements and let
  $w \colon V \to \{1,\ldots,W\}$ be element weights.
  In $O(Wn^{\omega})$ field operations we can find with high probability the maximum weight 
  of a common feasible set $F \in \cF_1 \cap \cF_2$.
  In particular, we can find the maximum cardinality of $|F|$ in
  $O(n^\omega)$ field operations. 
\end{restatable}

By self-reducibility, the search version can thus be solved by a 
factor $O(n)$ overhead.

\begin{corollary}
  \textsc{Weighted Linear DM Intersection}
  and \textsc{Weighted Perfect DM Parity}
  can be solved in $O(n^{\omega+1})$ field operations. 
\end{corollary}

Removing the overhead on this result appears significantly harder.
Indeed, a similar overhead exists for algorithms for weighted linear
matroid parity~\cite{CheungLL14}, and even removing the overhead for
\textsc{Weighted Perfect Matching} was significantly non-trivial~\cite{CyganGS15matching}.
We leave this as a (challenging) open question.

\subsubsection{Applications}

We now review some applications of the results. Not all of these
results are new, but they serve to demonstrate the applicability of
the setting.

One area of application is graph matching and factor problems. 
In the general factor problem, given a graph $G = (V, E)$ and a set of
integers $f(v) \subseteq \mathbb{N}$ for each vertex $v \in V$, we are
tasked with finding a spanning subgraph $H = (V, F)$ such that
$\deg_H(v) \in f(v)$ for every vertex $v \in V$.
Cornu{\'{e}}jols~\cite{Cornuejols88} showed that this problem is
polynomial-time solvable if each $f(v)$ has gaps of length at most 1
and NP-hard otherwise.
For a delta-matroid connection, for each $v \in V$ let $\delta(v)$
be the set of edges incident with $v$
and define $\cF_v=\{S \subseteq \delta(v) \mid |S| \in f(v)\}$.
Then the gap-1 condition is equivalent to $D_v=(\delta(v), \cF_v)$
forming a delta-matroid for every $v \in V$. We refer to $D_v$ as a
\emph{symmetric} delta-matroid. Then symmetric linear delta-matroids
correspond to cases where $f(v)=\{a,a+2,\ldots,b\}$,
so-called \emph{parity $(a,b)$-factors},
and in projected representation they additionally cover
$f(v)=\{a,a+1,\ldots,b\}$, so-called \emph{$(a,b)$-factors}. 

Consider the following setup. Let $G=(V,E)$ be a graph and let
$D_E$ be the matching delta-matroid on the ground set
$E_2=\{(e,v) \mid e \in E, v \in e\}$ with edges $\{(uv,u),(uv,v)\}$, $uv \in E$. 
Furthermore, let $V_f=\{(v,i) \mid v \in V, i \in \N, 1 \leq i \leq \max f(v)\}$.
Via the Ishikawa-Wakayama formula (see Lemma~\ref{lemma:cauchy-binet-ss}), we can construct a delta-matroid $D_f$
on ground set $V_f$ whose feasible sets correspond precisely to degree
sequences of subgraphs of $G$.
% (indeed, $D_f$ is the \emph{induced} delta-matroid
%from $D_E$ and the bipartite graph with edges
%$\{(v,i),(e,v)\}$~\cite{Bouchet89dam}; see Section~\ref{sec:delta-union}).
Thus, by imposing a second delta-matroid on the set $(v,i)$ for each $v \in V$, 
the \textsc{$(a,b)$-Factor} and \textsc{$(a,b)$-Parity Factor}
problems reduce to \textsc{Linear DM Intersection}.
This is similar to Gabow and Sankowski~\cite{GabowS21}.
There has also been recent work on \emph{weighted} general factors \cite{DudyczP17,Kobayashi23}.

For another example, let $G=(V,E)$ be a graph and $T \subseteq V$ a set of terminals.
Let $\cS$ be a partition of $T$. A \emph{$\cS$-path packing}
is a vertex-disjoint packing of paths where every path has
endpoints in distinct parts of $\cS$ and internal vertices disjoint
from $T$. A classical theorem of Mader shows a min-max theorem
characterizing the maximum number of paths in a $\cS$-path packing,
generalizing Menger's theorem and the Tutte-Berge formula;
see Schrijver~\cite{SchrijverBook}.
Wahlström~\cite{Wahlstrom24SODA} recently showed the following.
Let a set $F \subseteq T$ be \emph{feasible} if there is a $\cS$-path
packing whose set of endpoints is precisely $F$.
Then $D=(T,\cF)$ is a linear delta-matroid.
Then, via \textsc{Linear DM Intersection} as above,
we can solve the following ``$\cS$-factor'' problem:
Given $G$, $\cS$ and a prescribed set of degrees
$f(T_i)$, $T_i \in \cS$, is there a $\cS$-path packing in $(G,\cS)$
where precisely $f(T_i)$ paths have an endpoint in $T_i$?
The generalizations to $(a,b)$-factors and $(a,b)$-parity factors
can be handled similarly. 

%Possibilities:
%\begin{itemize}
%\item Linear matroid matching (comparison -- maybe related work instead) 
%\item Matching and some factor problems, possibly with weights
%\item ``Delta-sum of $n$ delta-matroids'', ``union of $n$
%  delta-matroids'', etc. -- artificial
%\item Max-weight intersection gives \ldots ?
%\item Shortest Mader paths?
%\item Boolean edge CSP with weights (linear constraint case)?
%\end{itemize}
%
%%for matching: I would talk about (1) symmetric delta-matroids, (2)
%%tractable cases of general factor
%%also a possibility to mention "interference free"  (maybe in related work)
%
%%missing mentions of path-packing problems and decorated or
%%constrained matching problems (I have some references) [Gallai's
%%construction, red-blue-style Gallai construction, Delta-Mader
%%maybe Bang-Jensen/Gutin
%%alternating cycle cover, etc.
%

\subsection{Related work}

For a survey on delta-matroids, see Moffatt~\cite{Moffatt19deltamatroids}.
In particular, as is well-surveyed by Moffatt,
there has been a recent resurgence of interest in delta-matroids
for studying graph embeddings in non-planar surfaces~\cite{ChunMNR19a,ChunMNR19b}.
Specifically, there is a class of delta-matroids that captures
so-called \emph{quasi-trees of ribbon graphs},
which generalizes graphic matroids to embedded graphs
in a way that also captures the structure of the embedding.
%The flexible structure of delta-matroid dualization appears to have
%been very useful in this study; see Moffatt~\cite{Moffatt19deltamatroids}
%for more.
In addition, delta-matroids have had applications
in structural matroid theory, via so-called \emph{twisted matroids}~\cite{GeelenGK00,Gollin2021Obstructions},
and in structural graph theory, by the close association between
the \emph{vertex minor} and \emph{pivot minor} operations on graphs,
and twists of delta-matroids~\cite{GeelenO09circle,Geelen97unimodular,Oum09deltagraphic}
(cf.~\cite{Oum17rankwidth}).

\emph{Algorithms on matroids.}
As mentioned, \textsc{Linear DM Parity}
and related problems generalize \textsc{Linear Matroid Parity}.
which in turn generalizes \textsc{Linear Matroid Intersection}.
% linear algebra
% Harvey? 
% gahd
The fastest algorithm for \textsc{Linear Matroid Parity} runs in $O(nr^{\omega-1})$ operations 
where $n$ is the size of the ground set and $r$ the rank~\cite{CheungLL14}.
The weighted case can be solved in $\tilde O(Wmr^\omega)$ where $W$
is the maximum weight of a pair~\cite{CheungLL14}, and was more
recently shown to be solvable in strongly combinatorial time $O(nr^3)$
via a complex algorithm of Iwata and Kobayashi~\cite{IwataK22SICOMP}.
For \textsc{Matroid Intersection} and other problems tractable
on general matroids, there is an active pursuit of efficient algorithms
whose running time is measured in terms of the number of oracle queries
over the matroid~\cite{ChakrabartyLS0W19,Blikstad21,BlikstadMNT23}.

\emph{Path-packing problems.}
In the \textsc{Shortest Disjoint $\cS$-Paths} problem, the task
is to find a $\cS$-path packing of $k$ paths with shortest total
length in the presence of non-negative edge weights.
Yamaguchi~\cite{Yamaguchi16} reduces this to \textsc{Weighted Linear Matroid Parity}.
$\cS$-path packings have been generalized to non-zero paths in
group-labelled graphs~\cite{ChudnovskyGGGLS06,ChudnovskyCG08}
and so-called \emph{non-returning} paths~\cite{Pap07,Pap08},
with some connection to matroid representations.
See also~\cite{Yamaguchi14,TanigawaY16}.
%also exist two Chudnovsky, two Pap

%\emph{General factors.}
%In the general factor problem, given a graph $G = (V, E)$ and a set of integers $f(v) \subseteq \mathbb{N}$ for each vertex $v \in V$, we are tasked with finding a spanning subgraph $H = (V, F)$ such that $\deg_H(v) \in f(v)$ for every vertex $v \in V$. 
%Cornu{\'{e}}jols~\cite{Cornuejols88} showed that this problem is polynomial-time solvable if each $f(v)$ has gaps of length at most 1 and NP-hard otherwise.
%There has been more recent work on weighted general factors \cite{DudyczP17,Kobayashi23}.
%The runtime is rather big.

\emph{Boolean edge CSP and Boolean planar CSP.}
Delta-matroid parity problems also occur in the context of restricted
versions of the Boolean CSP problems. Feder~\cite{Feder01fanout}
considered the Boolean CSP over a constraint language $\Gamma$,
with the additional restriction that every variable occurs in only two
constraints. %(The variant that every variable occurs in \emph{at most}
%two constraints easily reduces to the version where every variable
%occurs in \emph{exactly} two constraints.~\cite{Feder01fanout})
This has been referred to as the \emph{Boolean edge CSP}~\cite{KazdaKR19},
as such instances can be represented by an undirected graph where
the edges are variables and vertices are constraints.
Feder showed (assuming $\Gamma$ contains constants) that this problem is as hard as the unrestricted CSP over $\Gamma$,
unless every constraint in $\Gamma$ is a delta-matroid. 
In the latter case, the Boolean edge CSP reduces to \textsc{DM Parity}
and although the complexity for this problem is open in general,\footnote{%
  Although \textsc{DM Parity} is intractable in general, the form created by this reduction
  is as a direct sum over constant-sized components; 
%  We note that although delta-matroid parity is intractable in the
%  oracle model and NP-hard for arbitrary compact representations,
%  the instances of delta-matroid parity that come from reductions from
%  the Boolean edge CSP have a special structure where the
%  delta-matroid $D$ is the direct sum of smaller individual
%  delta-matroids $D_i$, where every component delta-matroid $D_i$ is
%  provided as an explicit list of feasible sets (or even assumed to be
%  one-time fixed, e.g., if $\Gamma$ is finite).  
  hence the general intractability results do not apply.
}
many tractable cases are known. In particular, Kazda et al.~\cite{KazdaKR19}
showed that the problem is in P if every delta-matroid is \emph{even}.
We note that if the constraints form (projected) linear delta-matroids,
representable over a common field, then via \textsc{Weighted DM Parity}
we can find a feasible solution of maximum or minimum Hamming weight.

Another restriction where delta-matroids are surprisingly relevant is
the \emph{Boolean planar CSP}. Here, the CSP instances are restricted
to having planar incidence graphs. Similarly to above, let $\Gamma$
be Boolean constraint language, and assume that the unrestricted CSP
over $\Gamma$ is intractable (as otherwise there is no reason to study
special cases). Dvor{\'a}k and Kupec~\cite{DvorakK15} showed that
either the planar CSP over $\Gamma$ is intractable, or the
problem reduces to Boolean edge CSP for a constraint language $\Gamma'$
where every constraint is an even delta-matroid. Thus by Kazda et al.~\cite{KazdaKR19},
there is a dichotomy of Boolean planar CSP as being in P or NP-hard,
where the (new) tractable cases correspond to \textsc{DM Parity}.

\emph{Parameterized complexity.}
As mentioned, linear matroids and related tools have seen significant
applications in parameterized
complexity~\cite{KratschW14TALG,KratschW20JACM,FominLPS16JACM,EibenKW24SODA}.
%reference the two books?
Recently, Wahlström extended the kernelization aspect of these tools
to delta-matroids, and used it to show a sparsification result related
to Mader's path-packing problem in terminal networks~\cite{Wahlstrom24SODA}.
In another direction, Eiben et al.\ gave a general method
for constructing efficient FPT algorithms by combining multivariate
generating polynomials with linear matroid side constraints~\cite{EibenKW24SODA}.
Applying this to the Pfaffians of the linear delta-matroids constructed in this paper 
should give a number of immediate FPT consequences.
We defer deeper investigations into these connections to later research.

\section{Preliminaries} \label{sec:prelim}

For two sets $A, B$, we let $A \Delta B=(A \setminus B) \cup (B \setminus A)$
denote their symmetric difference. 

For a matrix $A$ and a set of rows $S$ and columns $T$, we denote by $A[S, T]$ the submatrix containing rows $S$ and columns $T$.
If $S$ contains all rows ($T$ contains all columns), then we use the shorthand $A[\cdot, T]$ ($A[S, \cdot]$, respectively).
The $n \times m$ zero matrix and the $n \times n$ identity matrix is denoted by $O_{n \times m}$ and $I_{n}$, respectively.
We often drop the subscript when clear from context.

\emph{Delta-matroids.}
A \emph{delta-matroid} is a pair $D=(V,\cF)$ where $V$ is a ground set
and $\cF \subseteq 2^V$ a collection of \emph{feasible sets},
subject to the rule
\[
\forall A, B \in \cF,\, x \in A \Delta B\, \exists y \in A \Delta B :
A \Delta \{x,y\} \in \cF.
\]
This is known as the \emph{symmetric exchange axiom}.
For $D=(V,\cF)$ we let $V(D)=V$ and $\cF(D)=\cF$. 
A delta-matroid is \emph{even} if all feasible sets have the same
parity. Note that in this case we must have $x \neq y$ in the
symmetric exchange axiom, although this does not necessarily hold in
general. 

A \emph{separation oracle} for a delta-matroid $D=(V,\cF)$ is an
oracle that, given a pair $(S,T)$ of disjoint subsets of $V$,
reports whether there is a set $F \in \cF$ such that
$S \subseteq F$ and $F \cap T = \emptyset$. If so, the pair $(S,T)$ 
is \emph{separable}. A delta-matroid is \emph{tractable} if it has a
polynomial-time separation oracle.  

For a delta-matroid $D=(V,\cF)$ and $S \subseteq D$, the
\emph{twisting} of $D$ by $S$ is the delta-matroid
$D \Delta S = (V, \cF \Delta S)$ where
$
  \cF \Delta S = \{F \Delta S \mid F \in \cF\}.
$
This generalizes some common operations from matroid theory. 
The \emph{dual} delta-matroid of $D$ is $D \Delta V(D)$.
For a set $S \subseteq V(D)$, the \emph{deletion} of $S$ from $D$ 
refers to the set system
$
  D \setminus S = (V \setminus S, \{F \in \cF \mid F \subseteq V \setminus S\}),
$
The \emph{contraction} of $S$ refers to
$
  D/S = (D \Delta S) \setminus S = (V \setminus S, \{F \setminus S \mid F \in \cF, S \subseteq F\}).
$

\emph{Skew-symmetric matrices.}
A square matrix $A$ is \emph{skew-symmetric} if $A = -A^T$. In the
case that $A$ is over a field of characteristic 2, we will
additionally assume that it has zero diagonal, unless stated otherwise.
For a skew-symmetric matrix $A$ with rows and columns indexed by a set $V = [n]$, 
the \emph{support graph} of $A$ is the graph $G=(V, E)$ where $E = \{ uv \mid A[u, v] \ne 0 \}$.
A fundamental tool for working with skew-symmetric matrices is the
\emph{Pfaffian}, defined for a skew-symmetric matrix $A$ as
\begin{align*}
  \Pf A = \sum_{M} \sigma(M) \prod_{e \in M} A[u, v],
\end{align*}
where $M$ ranges over all perfect matchings of the support graph of $A$ and 
$\sigma(M) \in \{1,-1\}$ is the sign of the permutation:
\[
  \begin{pmatrix}
    1 & 2 & \cdots n - 1 & n \\
    v_1 & v_1' & \cdots v_{n/2} & v_{n/2}'
  \end{pmatrix},
\]
where $M = \{ v_i v_i' \mid i \in [n/2] \}$ with $v_i < v_i'$ for all $i \in [n/2]$.
It is known that $\det A=(\Pf A)^2$, hence in particular $\Pf A \neq 0$
if and only if $A$ is non-singular. However, for many algorithms it
will be more convenient to work directly with the Pfaffian. 
% For $v \in V$, we can give the ``Laplace expansion'' as follows:
% \begin{align*}
%   \Pf A = \sum_{v' \in V \setminus \{ v \}} \pm A[v, v'] \cdot \Pf \widehat{A}_{v, v'},
% \end{align*}
% where $\widehat{A}_{v, v'}$ is the matrix where the rows and columns indexed by $v$ and $v'$ are deleted.
% It follows that if a matrix $A'$ is obtained by multiplying every element in the row and column indexed by $v$  by a scalar $x_v$, then $\Pf A' = x_v \cdot \Pf A$.
In fact, Pfaffian generalizes the notion of determinants as follows.
\begin{lemma} \label{lemma:det-pf}
For an $n \times n$-matrix $M$, it holds that
$ \det M = (-1)^{n(n-1)/2} \Pf \begin{pmatrix}
  O & M \\
  -M^T & O \\
\end{pmatrix}. $
\end{lemma}

An important operation on skew-symmetric matrices is \emph{pivoting}.
Let $A \in \F^{n \times n}$ be skew-symmetric and let $S \subseteq [n]$ 
be such that $A[S]$ is non-singular.  Order the rows and columns of
$A$ so that
\[
  A =
  \begin{pmatrix}
    B & C \\
    -C^T& D
  \end{pmatrix},
\]
where $A[S]=B$.
Then the \emph{pivoting} of $A$ by $S$ is
\[
  A * S =
  \begin{pmatrix}
    B^{-1} & B^{-1}C \\
    CB^{-1}& D + C^T B^{-1} C
  \end{pmatrix}.
\]
Note that this is a well-defined, skew-symmetric matrix.

\begin{lemma}[\cite{tucker1960combinatorial}] \label{lemma:tucker}
It then holds, for any $X \subseteq [n]$, that
$
  \det (A*S)[X] = \frac{\det A[X \Delta S]}{\det A[S]},
$
In particular,
$(A*S)[X]$ is non-singular if and only if $A[X \Delta S]$ is
non-singular. 
\end{lemma}

Finally, let us note a formula on the Pfaffian of a sum of two skew-symmetric matrices:
\begin{lemma}[{\cite[Lemma 7.3.20]{Murota99}}]
  \label{lemma:sum-pf}
  For two skew-symmetric matrices $A_1$ and $A_2$ both indexed by $V$,
  we have
  \[ \Pf \, (A_1 + A_2) = \sum_{U \subseteq V} \sigma_U \Pf A_1[U] \cdot \Pf A_2[V \setminus U], \]
  where $\Pf A_i[\emptyset] = 1$ for $i = 1, 2$ and $\sigma_U \in \{1,-1\}$ is a sign of the permutation
  \begin{align*}
    \begin{pmatrix}
      1 & 2 & \cdots & |U| & |U| + 1 & \cdots & |V| - 1 & |V| \\
      u_1 & u_2 & \cdots & u_{|U|} & v_1 & \cdots & v_{|V \setminus U|-1} & u_{|V \setminus U|}
    \end{pmatrix},
  \end{align*}
  where $u_i$ and $v_i$ are the $i$-th largest elements of $U$ and $V \setminus U$, respectively.
\end{lemma}

The following is a generalization of the Cauchy-Binet formula to skew-symmetric matrices.
The algebraic approach of Lov\'asz \cite{Lovasz80matroid} for matroid parity can be derived from this formula (see \cite{MatoyaO22}).

\begin{restatable}[Ishikawa-Wakayama formula \cite{IshikawaW95}]{lemma}{iwformula}
  \label{lemma:cauchy-binet-ss}
  For a skew-symmetric $2n \times 2n$-matrix $A$ and a $2k \times 2n$-matrix $B$ with $k \le n$, we have
  \begin{align*}
    \Pf B A B^T = \sum_{U \in \binom{[2n]}{2k}} \det B[\cdot, U] \Pf A[U].
  \end{align*}
\end{restatable}

\emph{Linear representation.}
A skew-symmetric matrix defines a delta-matroid as follows.
For a skew-symmetric matrix $A \in \F^{V \times V}$ over a field $\F$, define
$ \cF = \{ X \subseteq V \mid A[X] \text{ is nonsingular}\} $. 
Then, $(V, \cF)$, which is denoted by $\bD(A)$, is a delta-matroid. 
We say that a delta-matroid $D = (V, \cF)$ is \emph{representable} over $\F$ if there is a skew-symmetric matrix $A \in \F^{V \times V}$ and a twisting set $X \subseteq V$ such that $D = \bD(A) \Delta X$.
If $A[X]$ is nonsingular, or equivalently $\emptyset \in \cF(D)$, we say that $D$ is \emph{directly representable} over $\F$.
Note that a directly representable delta-matroid $D$ can be represented without a twisting set $X$, as $\bD(A) \Delta X = \bD(A * X)$.
We will say that $D$ is \emph{directly represented} by $A$ if $D = \bD(A)$.
A delta-matroid is called \emph{normal} if $\emptyset$ is feasible. 
%Reference: Bouchet and Schwärzler~\cite{BouchetS98deltasum}, also used by Moffatt~\cite{Moffatt19deltamatroids}. 
Note that every linear delta-matroid is even, and that a linear delta-matroid is directly representable if and only if it is normal.
Linear delta-matroids are tractable~\cite{Bouchet95}.

In addition, we consider \emph{projected} linear delta-matroids. 
Let $D=(V,\cF)$ be a delta-matroid and $X \subseteq V$.
Then the projection $D|X$ is defined as $D|X=(V \setminus X, \cF|X)$
where $\cF|X = \{F \setminus X \mid F \in \cF\}$.
Then $D|X$ is a delta-matroid, although it is in general not even,
hence not linear. When $D$ is linear, then we refer to $D'=D|X$
as a \emph{projected linear delta-matroid}.
When $|X|=1$ we refer to this as an \emph{elementary projection}, 
following Geelen et al.~\cite{GeelenIM03}. 

As noted above, when $\F$ is not of characteristic 2, we also assume that $A$ has
a zero diagonal (which of course follows from the definition over all
fields not of characteristic 2). However, if $\F$ is a field of characteristic 2, 
then linear delta-matroids over $\F$ with a non-zero diagonal correspond
to a projected linear delta-matroids over $\F$; see Geelen et al.~\cite{GeelenIM03}.

For a matroid $M = (V, \cI)$ with the basis family $\mathcal B$, $D = (V, \mathcal B)$ is a delta-matroid.
If $M$ is represented by $A$, $D$ can be represented as follows.
Fix a basis $B \in \mathcal{B}$.
We may assume w.l.o.g.\ that $A[\cdot, B] = I$.
Define
\begin{align*}
  A' = \kbordermatrix{& B & V \setminus B \\ B & O & A[\cdot, V \setminus B] \\ V \setminus B & -A^T[V \setminus B, \cdot] & O}.
\end{align*}
Observe that for every $F \subseteq V$, $A'[F]$ is nonsingular if and only if $A[F \cap B, F \setminus B]$ is nonsingular.
Since $A[\cdot, B] = I$, this is equivalent to $A[\cdot, (B \setminus F) \cup (F \setminus B)]$ being nonsingular, and thus $D = \bD(A') \Delta B$.

Conversely, let $D=(V,\cF)$ be a delta-matroid. Then the set of
maximum-cardinality feasible sets in $D$ forms the set of bases of a
matroid $M=(V,\cI)$. Furthermore, if $D=\bD(A)$ is directly represented,
then $A$ (as a column space) is also a representation of the matroid $M$.
This is because if $B$ is a column basis for a skew-symmetric matrix $A$, then $A[B]$ is non-singular (see e.g., \cite{RabinV89} or \cite[Proposition 7.3.6]{Murota99}).

To avoid intricate representation issues, we assume that every
linear representation is given over some finite field.
We note that a representation over the rationals can be efficiently transformed
into an equivalent representation over a finite field.

\emph{Approximate linear representation.}
For a delta-matroid $D = (V, \cF)$, we say that a delta-matroid $D' = (V, \cF')$ is an \emph{$\varepsilon$-approximate representation} of $D$ if
% \[ \forall F' \in \cF', F' \in \cF \]
$\cF' \subseteq \cF$ 
and for every $F \in \cF$, the probability that $F \in \cF'$ is at least $1 - \varepsilon$.
For constructing an $\varepsilon$-approximate linear representation,
the Schwartz-Zippel lemma \cite{Schwartz80,Zippel79} (also referred to as the DeMillo-Lipton-Schwartz-Zippel lemma) comes in handy.
It states that a polynomial $P(X)$ of total degree at most $d$ over a field $\F$ becomes nonzero with probability at least $1 - d / |\F|$ when evaluated at uniformly chosen elements from $\F$, unless $P(X)$ is identically zero.

Let $G=(V,E)$ be an undirected graph and let $\cF \subseteq 2^V$
contain all sets $F \subseteq V$ such that $G[F]$ has a perfect matching. 
Then $D(G)=(V,\cF)$ is a delta-matroid referred to as the
\emph{matching delta-matroid} of $G$. 
The Tutte matrix gives rise to an approximate linear representation.
Note that setting $\varepsilon = O(2^{-|V|})$ (or lower) gives a matrix of polynomial 
size which with high probability is a correct representation of $D(G)$.
However, this will inflate the time needed for field operations over $\F$
by at least $\Omega(n)$, so for efficiency reasons we work with
$\varepsilon$-approximate representations where $\varepsilon$ is a parameter.

\begin{lemma} \label{lemma:tutte-representation}
  Let $G = (V, E)$ be a graph on $n$ vertices and $\F$ be a field with at least $n \cdot \lceil 1 / \varepsilon \rceil$ elements.
  We can construct a $\varepsilon$-approximate linear representation of the matching delta-matroid of $G$ over $\F$.
\end{lemma}
\begin{proof}
  Let $A$ be the Tutte matrix of $G$, where every edge variable is substituted with a uniformly randomly chosen element from $\F$.
  Fix $S \subseteq V$.
  If $G[S]$ has no perfect matching, then $\Pf A[S] = 0$.
  Otherwise, $\Pf A[S]$ is a nonzero polynomial of degree at most $n$.
  Thus, by the Schwartz-Zippel lemma, $A[S]$ is nonzero with probability at least $1 - \varepsilon$.
\end{proof}

%\paragraph{Delta-matroid polynomials.} Naturally, we also have
%polynomials such as the following:
%\begin{enumerate}
%\item Feasible sets in delta-matroids
%\item Ditto for delta-matroid intersection, or matroid/delta-matroid
%  intersection, and possibly also for twisted delta-matroids
%\item \emph{Delta-matroid partition}, e.g., given $\ell$ linear
%  delta-matroids over the same ground set and field we can enumerate
%  partitions of the ground set into sets that are feasible in each $D_i$
%\end{enumerate}
%\ldots but that is presumably to be handled (as appropriate) in the
%rest of the paper. 

% \todo[inline]{MW: Somewhere: (DeMillo-Lipton-)Schwartz, Zippel. Also: Tutte matrix. Query: Delta-Mader?}

% \todo[inline]{MW: TODO: Explicit lemma about the detailed success probability
  % of randomly constructing a representation of a Tutte matrix.
% }

\emph{Operations in matrix product time.}
Determinant, rank, basis, inverse can be found in $O(n^{\omega})$ time.
Given an $n \times 2n$-matrix, its row echelon form can be computed in $O(n^{\omega})$.
We can find a lexicographically smallest column basis in $O(n^{\omega})$ time.
See~\cite{von2013modern}.

\section{Contraction representation of linear delta-matroids} \label{sec:representation}

\newcommand{\repname}{contraction\xspace}

In this section, we introduce a novel linear representation for delta-matroids, called \emph{\repname{} representation}.
For the sake of clarity, we will say that the representation of a delta-matroid as $D = \bD(A) \Delta S$ is a \emph{twist representation}.
As we will see in Section~\ref{sec:enum-polys}, the \repname{} representation is useful in the design of more efficient algorithms for linear delta-matroids. 
We also give further results, supported by the new representation.
First, we show that the union and delta-sum of linear delta-matroids is linear (Sections~\ref{sec:delta-union} and~\ref{sec:delta-sum}).
Previously, this was only known to define delta-matroids~\cite{Bouchet89dam,BouchetS98deltasum}.
Next, we use this to provide a compact representation of projected linear delta-matroids
(Section~\ref{sec:projection}).
%an operation introduced by Kakimura and Takamatsu~\cite{KakimuraT14sidma}.\todo{TK: projection defined in \cite{BouchetC95}?}
All of these additional results will be useful in our algorithms.

\subsection{Contraction representations} \label{sec:contraction}
For a delta-matroid $D = (V, \cF)$, a \emph{\repname{} representation} of $D$ is a pair $(A,T)$ where $A$ is a skew-symmetric matrix over a field $\F$ whose rows and columns are labelled by $V \cup T$, such that $D = \bD(A) / T$, i.e.,
for every $F \subseteq V$, $F$ is feasible in $D$ if and only if $F \cup T$ is feasible in $\bD(A)$.
This is closely related to \emph{strong maps} of delta-matroids. For two delta-matroids $D$ and $D^\circ$,
$D^\circ$ is a \emph{strong map} of $D$ if there exists a delta-matroid $D^+=(V \cup Z,\cF)$
such that $D=D^+ \setminus Z$ and $D^\circ=D^+/Z$ (see Geelen et al.~\cite{GeelenIM03}).
Hence, if $D=(V,\cF)=\bD(A)/T$ is a \repname representation of a delta-matroid $D$, 
then $D$ is a strong map of the directly representable delta-matroid $\bD(A[V])$. 
We show that the \repname and twist representations are equivalent.

\begin{lemma} \label{lemma:representation}
  Given a delta-matroid $D$ in twist representation, 
  we can construct a contraction representation $D=\bD(A)/T$ of $D$
  deterministically in $O(n^2)$ time. 
\end{lemma}
\begin{proof}
  Let $D=(V,\cF)$ be given as $D=\bD(A) \Delta S$, $S \subseteq V$.
  For a set $T$ of size $|S|$, define a skew-matrix $A'$ over $V \cup T$ by
  \begin{align*}
    A' = \kbordermatrix{
    & V \setminus S & S & T \\
    V \setminus S & A[V \setminus S] & A[V \setminus S, S] & O \\
    S  & A[S, V \setminus S] & A[S] & I \\
    T & O & -I & O 
  },
  \end{align*}
  where $I$ is an identity matrix. Note that the support graph of $A'[S \cup T]$ 
  has a unique perfect matching (namely, every vertex in $T$ has degree one),
  thus $\Pf A'[S \cup T] = \pm 1$ and $A'[S \cup T]$ is non-singular. 
  Thus, we can construct the matrix $A^* = A' * (S \cup T)$.
  Note that
  \begin{align*}
    (A[S \cup T])^{-1} = \begin{pmatrix}
      A[S] & I \\ -I & O
    \end{pmatrix}^{-1} =
    \begin{pmatrix}
      O & -I \\ I & A[S] 
    \end{pmatrix},
  \end{align*}
  and consequently, the result of pivoting is
  \begin{align*}
    A^* = \kbordermatrix{
    & V \setminus S & S & T \\
    V \setminus S & A[V \setminus S] & O & A[V \setminus S, S] \\
    S  & O & A[S] & I \\
    T & A[S, V \setminus S] & -I & O 
    }.
  \end{align*}
  Clearly, $A^*$ can be constructed in $O(n^2)$ time.
  Now by Lemma~\ref{lemma:tucker}, for any $F \subseteq V$, $A^*[F \cup T]$
  is non-singular if and only if $A'[(F \cup T) \Delta (S \cup T)]=A'[F \Delta S]$ is.
  Since $F \Delta S \subseteq V$ and $A'[V]=A[V]$, 
  this is equivalent to $F \in \cF(D)$, thereby showing that $\bD(A^*)/T$ is a contraction representation of $D$. 
\end{proof}

\begin{lemma} \label{lemma:contra-to-twist}
  Given a \repname{} representation $D = \bD(A) / T$, we can find a
  twist representation of $D$ deterministically using $\Oh(n^{\omega})$ field operations.
\end{lemma}
\begin{proof}
 Let $S \subseteq V$ be a set such that $A[S \cup T]$ is non-singular;
 such a set exists since $\cF(D) \neq \emptyset$ by assumption,
 and can be found efficiently over $A$. Then $A' = A * (S \cup T)$ is well-defined.
 Let $A^* = A' \setminus T$.
 Observe that $D = \bD(A) \Delta S$, that is, for every $F \subseteq V$, $A^*[F \Delta S] = A'[F \Delta S]$ is nonsingular if and only if $A[(F \Delta S) \Delta (S \cup T)] = A[F \cup T]$ is nonsingular by Lemma~\ref{lemma:tucker}.
 All operations above can be performed in matrix multiplication time.
\end{proof}

We also observe that the contracted set $T$ in a representation of a
delta-matroid $D=(V,\cF)$ never needs to be larger than $|V|$.

\begin{lemma} \label{lemma:contraction-reduce}
  Given a \repname representation $D=\bD(A)/T$ of a delta-matroid $D=(V,\cF)$, 
  in $O(n^\omega)$ field operations, where $n=|V|+|T|$,
  we can find a \repname representation $D=\bD(A')/T'$ where $|T'| \leq |V|$. 
\end{lemma}
\begin{proof}
 We claim that if $|T| > |V|$ then there is a feasible set $F_0 \subseteq T$.
 Let $F$ be any feasible set in $D$ and let $F'=F \cup T$.
 Let $D'=\bD(A)$. Then $D'$ is a linear delta-matroid and $F'$ 
 is feasible in $D'$. Furthermore $\emptyset$ is feasible in $D'$
 since $D'$ is directly represented. Assume $|T| > |V|$ as otherwise
 there is nothing to do. Consider the following process:
 Let $v \in F' \cap V$. By the symmetric exchange axiom, there is
 an element $v' \in F' \Delta \emptyset=F'$ such that $F' \Delta \{v,v'\}$
 is feasible. Update $F'\gets F' \Delta \{v,v'\}$. Since $F'$ strictly
 shrinks at every step, this process terminates with a feasible set
 $F_0 \subseteq T$, of cardinality at least $|T|-|V|$.
 Now, given that such a set exists, we can compute
 a column basis $B$ of $A[T]$, which will be feasible in $A$,
 and construct $A'=(A*B) \setminus B$, $T' = T \setminus B$.
 Clearly, $D=\bD(A')/T'$ is also a \repname representation.
 Furthermore, both steps can be performed in $O(n^\omega)$
 field operations. 
\end{proof}

\paragraph*{Example for matroids.} For illustration purposes, let us provide 
an explicit representation of the bases of a linear matroid
as the feasible sets of a delta-matroid in contraction representation.

Let $M = (V, \cI)$ be a matroid represented by a matrix $A \in \F^{k \times V}$,
and let $\cB$ be the set of bases of $M$. 
We can give a contraction representation for the delta-matroid $D = (V, \cB)$ as follows:
\begin{align*}
D=\bD(A')/T \text{ where } A'=  \kbordermatrix{& T & V \\ T & O & A \\ V & -A^T & O}
\end{align*}
Note that $A'[\cdot, B]$ is non-singular if and only if $A'[T \cup B]$ is non-singular by Lemma~\ref{lemma:det-pf}.
For a basis $B \in \cB$, applying the proof of Lemma~\ref{lemma:contra-to-twist} gives a twist representation:
\begin{align*}
  \kbordermatrix{& B & V \setminus B \\ B & O & (A[B])^{-1} A[\cdot, V \setminus B] \\ V \setminus B & -((A[B])^{-1} A[\cdot, V \setminus B])^T & O}
\end{align*}
We recover the known twist-representation (see Section~\ref{sec:prelim}) when $A[B] = I$.

In addition, $\bD(A')|T$ is a projected linear representation of the independent set delta-matroid of $M$.

\subsection{Delta-matroid union} \label{sec:delta-union}

We next consider an immediate way to combine two delta-matroids into a
new delta-matroid, the \emph{delta-matroid union} (surveyed in the introduction).
Let $D_1=(V_1,\cF_1)$ and $D_2=(V_2,\cF_2)$ be two
delta-matroids on not necessarily disjoint ground sets and let $V=V_1 \cup V_2$.
Define $\cF = \cF_1 \uplus \cF_2 := \{F_1 \cup F_2 \mid F_1 \in \cF_1, F_2 \in \cF_2, F_1 \cap F_2 = \emptyset\}$
as the collection of sets that can be produced as disjoint unions from
$\cF_1$ and $\cF_2$, and write $D=(V,\cF)=D_1 \cup D_2$.
Then Bouchet~\cite{Bouchet89dam} showed that $D=(V,\cF)$ is a delta-matroid.
We show that furthermore, if $D_1$ and $D_2$ are linear or projected linear
then so is $D$, and an $\varepsilon$-approximate representation can be
constructed in polynomial time.
Note that we may as well assume that $V=V_1=V_2$, by adding the
missing elements to the respective delta-matroid as loops.

\begin{lemma} \label{lemma:union}
  Let $D_1=(V,\cF_1)$ and $D_2=(V,\cF_2)$ be linear or projected
  linear delta-matroids defined over a common field $\F$
  and given in contraction representation.
  Then the delta-matroid union $D=D_1 \cup D_2$
  is a linear (respectively projected linear) delta-matroid,
  and an $\varepsilon$-approximate representation of $D$
  can be constructed in $O(n^2)$ field operations
  over an extension field of $\F$ with at least
  $n \cdot \lceil 1/\varepsilon \rceil$ elements.  
\end{lemma}
\begin{proof} 
 Let $D_1=(\bD(A_1)/T_1)|X_1$ and $D_2=(\bD(A_2)/T_2)|X_2$ be the
 representations of $D_1$ and $D_2$, where $X_1$, $X_2$ may be empty
 (if $D_1$, $D_2$ are linear) and $T_1$, $T_2$, $X_1$ and $X_2$ are all pairwise disjoint. 
 Let $V=[n]$ (w.l.o.g.) and define a set of indeterminate
 variables $y_1, \ldots, y_n$.
 Define a matrix $A$ indexed by $V \cup T_1 \cup T_2 \cup X_1 \cup X_2$
 such that $A[V \cup T_1 \cup X_1]=A_1$ and $V[V \cup T_2 \cup X_2]=A_2$
 except that for $i, j \in V$ we have $A[i,j] = A_1[i,j] + y_iy_jA_2[i,j]$. 
 All remaining blocks of $A$ are zero.
 We claim that $D=(\bD(A)/(T_1 \cup T_2))|(X_1 \cup X_2)$. 
 Indeed, let $F=F_1 \cup F_2$ for $F_1 \in \cF_1$
 and $F_2 \in \cF_2$. Then there are $S_1 \subseteq X_1$
 and $S_2 \subseteq X_2$ such that
 $\Pf A_i[F_i \cup S_i \cup T_i] \neq 0$ for $i=1, 2$.
 By Lemma~\ref{lemma:sum-pf}, $\Pf A[F \cup S_1 \cup S_2 \cup T_1 \cup T_2]$
 has a term which up to a sign term contributes
 $\prod_{i \in F_2} y_i \Pf A_1[F_1 \cup S_1 \cup T_1] \cdot \Pf A_2[F_2 \cup S_2 \cup T_2]$
 which is non-zero by assumption. Furthermore, the monomial $\prod_{i \in F_2} y_i$
 is not contributed by any other term. Hence the Pfaffian is non-zero
 as a polynomial in $y_i$.
 
 Conversely, if $\Pf A[F \cup S_1 \cup S_2 \cup T_1 \cup T_2] \neq 0$
 for some $F \subseteq V$, $S_1 \subseteq X_1$, $S_2 \subseteq X_2$
 then by Lemma~\ref{lemma:sum-pf} and by construction of $A$
 there is at least one partition $F=F_1 \cup F_2$ such
 that $\Pf A_i[F_i \cup S_i \cup T_i] \neq 0$ for $i=1, 2$,
 hence $F = F_1 \cup F_2$ is feasible in $D$.

 The running time is trivial, and the success probability is a
 straight-forward application of the Schwartz-Zippel lemma.  
\end{proof}

Bouchet presents delta-matroid union as a special case of a notion of a delta-matroid \emph{induced by}
a delta-matroid and a bipartite graph~\cite{Bouchet89dam}.
This is directly analogous to the notion of a \emph{linkage matroid};
see Oxley~\cite{OxleyBook2}. Via the Ishikawa-Wakayama formula,
letting $B$ be the Edmonds matrix of the bipartite graph,
it follows that a delta-matroid induced by a (projected) linear delta-matroid
is a (projected) linear delta-matroid, and an $\varepsilon$-approximate
representation can be constructed in matrix multiplication time.
We omit the details. 

\subsection{The delta-sum of linear delta-matroids}
\label{sec:delta-sum}

Let $D_1=(V_1, \cF_1)$ and $D_2=(V_2,\cF_2)$ be delta-matroids on not
necessarily disjoint ground sets. 
Let $V=V_1 \cup V_2$ and $\cF=\{F_1 \Delta F_2 \mid F_1 \in \cF_1, F_2 \in \cF_2\}$.
Then $D=(V,\cF)$ is called the \emph{delta-sum} $D=D_1 \Delta D_2$ of
$D_1$ and $D_2$, and is itself a delta-matroid.
Bouchet and Schwärzler~\cite{BouchetS98deltasum} give a proof,
citing unpublished work by Duchamp for the result.  

Interestingly, even though $D_1 \Delta D_2$ is always a delta-matroid, it is
not always tractable. Let us recall the definitions. 
A separation oracle for a delta-matroid $D=(V,\cF)$ is an oracle that reports,
for every pair $(S,T)$ with $S \cap T = \emptyset$ and $S, T \subseteq V$,
whether there exists a feasible set $F \in \cF$ such that $S \subseteq F$ and $F \cap T = \emptyset$.
A delta-matroid is tractable if its separation oracle can be evaluated in polynomial time.
In particular, linear delta-matroids are tractable.
However, Bouchet and Schwärzler note that
the delta-sum of two tractable delta-matroids is not in general tractable.
Indeed, let $(M, \Pi)$ be an instance of the \textsc{Matroid Parity} problem,
i.e., $M=(V,\cI)$ is a matroid and $\Pi$ is a partition of $V$ into pairs, 
where the question is if there is a basis $B$ of $M$ such that $B$ is a union of pairs.
Let $D_\Pi$ be the matching delta-matroid of the graph with edge set $\Pi$.
Recall that the bases of $M$ form a delta-matroid; this delta-matroid has a tractable separation oracle
if and only if $M$ has a tractable independence oracle. 
Then $\emptyset$ is feasible in $M \Delta D_\Pi$ if and only if $(M,\Pi)$
is a yes-instance, which is known to require an exponential number of queries
in the worst case and is intractable even when $M$ is given explicitly \cite{Lovasz80matroid}.

We show that the delta-sum of linear delta-matroids is linear, and thereby tractable,
in the case that $D_1$ and $D_2$ are given as representations over a common field. 
This fits well with the fact that \emph{linear} matroid parity is in P \cite{Lovasz80matroid}.
By Lemma~\ref{lemma:representation}, we can work with \repname representations. 
%The construction is randomized; as in the construction of a Tutte matrix (Lemma~\ref{lemma:tutte-representation}),
%we take a parameter $\varepsilon > 0$ to control the error rate.
%Again, choosing $\varepsilon = O(1/2^n)$ yields a correct representation of $D_1 \Delta D_2$ with high probability
%at the cost of field operations over $\F'$ becoming more expensive by a factor of $\tilde O(n)$. 
%
\begin{lemma}
  \label{lemma:delta-sum}
  The delta-sum of (projected) linear delta-matroids over a common field is a (projected) linear
  delta-matroid, and a representation can be computed in randomized polynomial time.
  More precisely, let $D_1$ and $D_2$ be linear delta-matroids given
  in contraction representation over a common finite field $\F$.
  Let $\varepsilon > 0$ be given and let $\F'$ be a field extension of $\F$
  with at least $n \cdot \lceil 1/\varepsilon \rceil$ elements.
  We can construct an $\varepsilon$-approximate contraction
  representation of $D_1 \Delta D_2$
  in $O(n^\omega)$ field operations over $\F'$. 
\end{lemma}
\begin{proof}
  Let $D_1=(V,\cF_1)$ and $D_2=(V,\cF_2)$; assume w.l.o.g.\ that 
  they are over the same ground set (by adding dummy elements if necessary).
  For the projected case, assume
  that $D_1=D_1^+|X_1$ and $D_2=D_2^+|X_2$ for disjoint sets $X_1$, $X_2$;
  we can then construct $D_1 \Delta D_2 = (D_1^+ \Delta D_2^+)|(X_1 \cup X_2)$.  
  Thus we focus on the linear case and let
  $D_1=\bD(A_1)/T_1$ and $D_2=\bD(A_2)/T_2$ for skew-symmetric matrices 
  $A_i \in \F^{(V \cup T_i) \times (V \cup T_i)}$, $i=1, 2$.
  We first show a representation of $D_1 \Delta D_2$ in three ``layers''.  
  Let $V_1$ and $V_2$ be disjoint copies of $V$, and for $v \in V$ let $v^i$ denote its copy in $V_i$. 
  Let $V^+=V \cup V_1 \cup V_2$, and let $H$ be a graph over $V^+$
  consisting of the disjoint union of triangles $\{v, v^1, v^2\}$ over all $v \in V$. 
  Let $A_H$ be the Tutte matrix of $H$, and let $A'$ be the disjoint union
  of $A_1$ over $V_1 \cup T_1$ and $A_2$ over $V_2 \cup T_2$, i.e.,
  \begin{align*}
    A_H = \kbordermatrix{
        & V  & V_1  & T_1  & V_2  & T_2 \\
      V & O  & B_1  &  O   & B_2  & O  \\
      V_1&-B_1& O   &   O  &  B_3  &  O \\
      T_1&  O & O   &  O   &   O  &  O \\
      V_2&-B_2& -B_3  &  O   &   O  &  O \\
      T_2&  O & O   &  O   &   O  &  O 
    }
  \end{align*}
  where $B_1$, $B_2$ and $B_3$ are diagonal matrices whose entries are all distinct independent variables. 
  With the same partition, $A'$ is written as
  \begin{align*}
    A' = \kbordermatrix{
        & V  & V_1  & T_1  & V_2  & T_2 \\
      V & O  & O    &  O   & O    & O  \\
      V_1& O  & A_1[V_1] & A_1[V_1,T_1] &   O    &  O \\
      T_1&  O & A_1[T_1,V_1] &  A_1[T_1]   &   O  &  O \\
      V_2& O  & O & O & A_2[V_2] &  A_2[V_2,T_2] \\
      T_2&  O & O & O & A_2[T_2,V_2]  &  A_2[T_2]
    }.
  \end{align*}
  Then let $A=A_H+A'$. 
  For $S \subseteq V$, let $S^+=S \cup V_1 \cup T_1 \cup V_2 \cup T_2$.
  We claim that for $S \subseteq V$, $A[S^+]$ is non-singular if and only if $S$ is feasible in $D_1 \Delta D_2$.
  We note that by Lemma~\ref{lemma:sum-pf}, for any set $S$ we have
  $\Pf A[S^+] = \sum_{U \subseteq S^+} \sigma_U \Pf A_H[U] \cdot \Pf A'[S^+ \setminus U]$
  for some sign terms $\sigma_U$.
  Furthermore, since $A_H$ is a Tutte matrix, every matching in $H$ contributes
  an algebraically distinct term; thus $A[S]$ is non-singular
  if and only if there exists a set $U \subseteq S^+$
  such that $H[U]$ has a perfect matching and $A'[S^+ \setminus U]$ is non-singular. 
  
  First, assume that there are $F_1 \in \cF_1$ and $F_2 \in \cF_2$
  such that $F_1 \Delta F_2 = S$, $S \subseteq V$.
  Consider 
  \[
    U = S^+ \setminus
    \Bigl(\{v^1 \mid v \in F_1\} \cup T_1 \cup
    \{v^2 \mid v \in F_2\} \cup T_2
    \Bigr).
  \]
  Then $H[U]$ consists of edges $vv^1$ for $v \in F_2 \cap S$,
  $vv^2$ for $v \in F_1 \cap S$ and $v^1v^2$ for $v \in V \setminus (F_1 \cup F_2)$.
  Thus $H[U]$ has a perfect matching. Furthermore $A'[S^+ \setminus U]$
  is just the diagonal block matrix with blocks $A_1[F_1 \cup T_1]$
  and $A_2[F_2 \cup T_2]$, which is non-singular by assumption. 
  
  Second, assume that $\Pf A[S^+] \neq 0$ for $S \subseteq V$, and let $U \subseteq S^+$
  be a term such that $\Pf A_H[U] \cdot \Pf A'[S^+ \setminus U] \neq 0$.
  Then $S \subseteq U \subseteq S \cup V_1 \cup V_2$
  and $U$ contains an even number of elements of every triangle $\{v,v^1,v^2\}$.
  In particular, every vertex $v \in S$ is matched against $v^1$ or $v^2$ in $H[U]$. 
  
  Let $F_1=V_1 \setminus U$ and $F_2=V_2 \setminus U$,
  and write $F_1'=\{v \in V \mid v^1 \in F_1\}$
  and $F_2'=\{v \in V \mid v^2 \in F_2\}$. 
  Then for every $v \in S$, precisely one of the statements
  $v^1 \in F_1$ and $v^2 \in F_2$ is true, i.e., $v \in F_1' \Delta F_2'$.  
  Furthermore, for every vertex $v \in V \setminus S$,
  either $v^1, v^2 \in U$ and $v$ is disjoint from $F_1' \cup F_2'$,
  or $v \in F_1' \cap F_2'$. Thus $S=F_1' \Delta F_2'$.
  Again, $A'[S^+ \setminus U]$ has a diagonal block structure
  with blocks $A_1[F_1 \cup T_1]$ and $A_2[F_2 \cup T_2]$,
  and since $\Pf A'[S^+ \setminus U] \neq 0$ by assumption
  we find $F_1 \in \cF_1$ and $F_2 \in \cF_2$. 
  Now $D_1 \Delta D_2=\bD(A)/(V_1 \cup T_1 \cup V_2 \cup T_2)$
  is a \repname representation of the delta-sum. 

  Regarding the constructive aspects of running time and success probability,
  the only non-deterministic step above is taking a representation of $A_H$. 
  For any $S \subseteq V$, $\Pf A[S^+]$ is a polynomial of degree at most $n$ in 
  the variables associated with the edges of $H$. Hence, let $\F'$ be
  an extension field of $\F$ with $|\F'| \geq n \lceil 1/\varepsilon \rceil$
  and create $A_H$ by replacing the edge variables by values from $\F'$
  chosen independently and uniformly at random. Then for any $S \subseteq V$,
  if $S$ is infeasible in $D_1 \Delta D_2$ then $A[S^+]$ is singular
  for every choice of $A_H$, and if $S$ is feasible then $A[S^+]$
  is non-singular with probability at least $1-|F'|/n \geq 1-\varepsilon$
  by the Schwartz-Zippel lemma. 
\end{proof}

If desired, we may reduce the contraction set $T=V_1 \cup T_1 \cup V_2 \cup T_2$
in this construction to a smaller set $T'$, $|T'| \leq |V|$, as per
Lemma~\ref{lemma:contraction-reduce}.

\subsection{Delta-matroid projections}
\label{sec:projection}

Recall that a projected linear delta-matroid $D=(V,\cF)$ is a delta-matroid
represented as $D=D'|X$ where $D'=(V \cup X, \cF')$ is a linear delta-matroid.
Projections of linear delta-matroids were studied by Geelen et al.~\cite{GeelenIM03}
in the context of linear delta-matroids over fields of characteristic 2, 
and by Kakimura and Takamatsu~\cite{KakimuraT14sidma} regarding
generalizations of constrained matching problems. 

We observe that if $D$ is linear, then the even (respectively odd)
sets of $D|X$ form a linear delta-matroid, and that every projected
linear delta-matroid $D|X$ can be represented via an elementary projection. 

\begin{lemma} \label{lemma:delta-project}
  Let $D=(V \cup X,\cF)$ be a linear delta-matroid.
  Then the following delta-matroids are linear and approximate representations can
  be constructed efficiently.
  \begin{enumerate}
  \item A linear delta-matroid $D'=(V \cup X', \cF')$ such that
    $D|X=D'|X'$ and $|X'| \leq 1$
  \item The delta-matroid $D_e=(V,\cF_e)$ where $\cF_e$ contains the
    sets of $\cF|X$ of even cardinality
  \item The delta-matroid $D_o=(V,\cF_o)$ where $\cF_o$ contains the
    sets of $\cF|X$ of odd cardinality
  \end{enumerate}
  More precisely, let $D=\bD(A)/T$ in contraction representation
  over a finite field $\F$ and let $\varepsilon > 0$ be given.
  Let $\F'$ be a field extension of $\F$ with at least $n \cdot \lceil
  1/\varepsilon \rceil$ elements. 
  We can construct an $\varepsilon$-approximate contraction
  representation of each of the above delta-matroids in $O(n^2)$
  operations over $\F'$.
\end{lemma}
\begin{proof}
 Let $D=\bD(A)/T$ where $A$ is a skew-symmetric matrix indexed by $V \cup X \cup T$.
 Assume $|X| > 1$ or there is nothing to show.  Let $D_K$ be the matching delta-matroid
 over $V \cup X$ for the graph which is a clique on $X$ and
 has no other edges. Let $D'=\bD(A')/T'$ be a contraction
 representation of the delta-matroid union $D \cup D_K$.
 Let $D'=(V \cup X,\cF')$ and select $x \in X$.
 We note that for any set $F \subseteq V$, we have $F \in \cF|X$ if
 and only if either $F \cup X \in \cF'$ or $(F \cup X)-x \in \cF'$.
 Indeed, on the one hand let $F \in \cF|X$ and let $F \cup S \in \cF$,
 $S \subseteq X$. If $|X \setminus S|$ is even, then $X \setminus S$ is feasible in $D_K$,
 hence $F \cup X$ is feasible in $D'$.
 If $|X \setminus S|$ is odd, then $X \setminus (S+x)$ is feasible in $D_K$,
 hence $(F \cup X)-x$ is feasible in $D'$.
 Conversely, if $F \cup X \in \cF'$ respectively $(F \cup X)-x \in \cF'$ then there
 exists $S \subseteq X$ such that $F \cup S \in \cF$ (and the remaining elements
 are feasible in $D_K$).
 Hence we represent the same delta-matroid $D|X$ by contracting $X-x$ and leaving $X'=\{x\}$. 
 
 Furthermore, the parity of $|S|$ is controlled by the
 parity of $|F|$ (since $D$ is even), hence deleting or contracting $x$
 produces the odd and even halves of $D|X$. 

 Regarding running time and success probability, there are two
 non-deterministic steps, namely the construction of $D_K$
 and the construction of $D \cup D_K$.
 Let $D_K'=\bD(A')$ be an $\varepsilon/2$-approximate representation
 of $D_K$, using Lemma~\ref{lemma:tutte-representation},
 and let $D'$ be an $\varepsilon/2$-approximate representation 
 of $D \cup D_K'$ using Lemma~\ref{lemma:union}.
 It can easily be verified that errors are one-directional,
 i.e., for any set $S$ infeasible in $D|X=(D \cup D_K)/(X-x)$,
 $S \cup (X-x)$ remains infeasible in $D'$. Now let $F$ be feasible
 in $D|X$ and let $F \cup S$ be feasible in $D$, for some $S \subseteq X$.
 With probability at least $1-\varepsilon/2$, $X \setminus S$
 respectively $X \setminus (S+x)$ is feasible in $D_K'$,
 hence $F \cup X$ or $(F \cup X)-x$ is feasible in $D \cup D_K'$.
 Subject to this, with a further probability of at least $1-\varepsilon/2$ 
 $F \cup X$ or $(F \cup X)-x$ is also feasible in $D'$.
 Note that these steps take $O(n^2)$ time. 
 Finally, the output representation can be produced by simply adding $X-x$
 to the contraction set in the representation of $D'$.
\end{proof}

\subsection{Delta-matroid parity min-max formula in contraction form}
\label{sec:dmp-min-max}

We recall the definition of the \textsc{DM Parity} problem. 
Let $D=(V,\cF)$ be a delta-matroid and $\Pi$ a partition of $V$ into pairs. 
For a set $F \subseteq V$, let $\delta_{\Pi}(F) = |\{ P \in \Pi : |F \cap P| = 1 \}|$
denote the number of pairs broken by $F$, and define
$\delta(D, \Pi) = \min_{F \in \cF} \delta_{\Pi}(F)$.
The goal of \textsc{DM Parity} is to find a set $F \in \cF$ with
$\delta_\Pi(F)=\delta(D,\Pi)$. 
As well as an algorithm for \textsc{Linear DM Parity}, 
Geelen et al.~\cite{GeelenIM03} showed a min-max theorem for the value
of $\delta(D,\Pi)$ that applies when $D$ is linear. 
However, its formulation is (in our opinion) in unfamiliar terms,
and even its validity as a lower bound is non-immediate. 
We observe an equivalent reformulation in terms of contraction
representation. This is of course mathematically equivalent to the
known result~\cite{GeelenIM03}, but in our opinion it is a
structurally much more transparent statement. 

In the following, when we say that $A[U]$ is a direct sum of $A[U_1]$,
$A[U_2]$, \dots we mean that $A[U]$ is a diagonal block matrix with
blocks $A[U_1]$, $A[U_2]$, \dots such that all other blocks are all-zero.

\begin{lemma} \label{lm:dmparity-minmax}
 Let $D=(V,\cF)$ be a linear delta-matroid and $\Pi$ a partition of $V$ into pairs.
 There is a contraction representation $D=\bD(A)/T$ of $D$
 for a skew-symmetric matrix $A$ and partitions $V=V_1 \cup \ldots \cup V_k$
 and $T=T_0 \cup T_1 \cup \ldots \cup T_k$ such that the following hold.
 \begin{enumerate}
 \item $A[V \cup (T \setminus T_0)]$ is a direct sum of $A[V_i \cup T_i]$ for $i=1, \ldots, k$
 \item For every pair $P \in \Pi$ there is an index $i \in [k]$
   such that $P \subseteq V_i$
 \item The number of parts $i \in [k]$ such that $|T_i|$ is odd equals $\delta(D,\Pi)+|T_0|$.
 \end{enumerate}
\end{lemma}
\begin{proof}
 The min-max statement of Geelen et al.~\cite{GeelenIM03} is defined
 as follows. Let $D^\circ$ be defined from $D$ by a strong map,
 i.e., there is a delta-matroid $D^+$ on ground set $V \cup Z$
 such that $D=D^+ \setminus Z$ and $D^\circ=D^+/Z$.
 Let $V=V_1 \cup \ldots \cup V_k$ be the most fine-grained direct sum
 decomposition of $D^\circ$ that is refined by $\Pi$, i.e., the partition
 is produced starting from the most fine-grained direct sum
 decomposition of $D^\circ$ and then merging components that are
 connected by a pair from $\Pi$. Let $\text{odd}(D^\circ,\Pi)$
 be the number of components $V_i$ in this decomposition such that
 the feasible sets of the induced delta-matroid $D^\circ(V_i)$ have
 odd cardinality. Then
 \[
   \delta(D, \Pi) = \max (\text{odd}(D^\circ,\Pi)-|Z|)
 \]
 where the max is taken over strong maps $D^\circ$ of $D$ as above.
 Here, $D^+$ is a linear delta-matroid (in fact, it is constructed
 during their algorithm).

 We begin with a ``mixed'' representation of $D$ as a contraction and
 a twist. Since $D^\circ$ is a delta-matroid it has at least one
 feasible set, hence there exists a 
 feasible set $F \in \cF(D^+)$ with $Z \subseteq F$.
 Thus we can represent $D^+$ as $D^+=\bD(A) \Delta S^+$ over some
 matrix $A$ such that $Z \subseteq S^+$. Let $S=S^+ \setminus Z$.
 Since $D=D^+ \setminus Z$, a set $F \subseteq V$ is feasible in $D$
 if and only if $F \Delta S^+ = (F \Delta S) \cup Z$ is feasible in $A$.
 Thus $D=(\bD(A)/Z) \Delta S$. 

 Similarly, $D^\circ$ is represented as
 $D^\circ=(\bD(A) \setminus Z) \Delta S =\bD(A[V]) \Delta S$,
 and since direct sum decomposition
 is invariant under twists, $V=V_1 \cup \ldots \cup V_k$ is a
 diagonal block decomposition of $A[V]$.
 Finally, a component $V_i$ is odd in $D^\circ$ if and only if
 $|S \cap V_i|$ is odd.
 By Geelen et al.~\cite[Theorem~3.1]{GeelenIM03},
 the number of odd components is therefore precisely $|Z|+\delta(D,\Pi)$.

 To complete into a pure contraction representation, we simply follow
 the proof of Lemma~\ref{lemma:representation}. For $i \in [k]$,
 let $S_i=S \cap V_i$. Create a new set of elements
 $T'=T_1 \cup \ldots \cup T_k$ where $|T_i|=|S_i|$ for each $i$.
 Define a matrix $A'$ indexed by $V \cup Z \cup T'$
 where $A'[V \cup Z]=A$ and $A'[S_i,T_i]$ induces an identity matrix
 for each $i \in [k]$ (and otherwise all new entries are zero).
 New consider $A^*=A*(S \cup T)$. From the form of $A^*$ given in the 
 proof of Lemma~\ref{lemma:representation}, the support graph of
 $A$ is unchanged by this pivot except (1) $S_i$ now only neighbours $T_i$
 and (2) $T_i$ neighbours $V_i$. Let $T=Z \cup T'$. Then $D=\bD(A^*)/T$,
 the direct sum decomposition of $A[V]$ transfers to $A^*[V \cup T']$ 
 and for each $i \in [k]$, $|S_i|$ is odd if and only if $|T_i|$ is odd.
\end{proof}

We wish to point out the structural similarity of this statement to
the Tutte-Berge formula for graph matching. Indeed, let $D=\bD(A)/T$
as in the lemma and let $G$ be the support graph of $A$.
Let $G'$ be $G$ with edges $\Pi$ added. 
Then $G'-T_0$ has $|T_0|+\delta(D,\Pi)$ odd components, and the
matching deficiency of $G'$ is (at least) $\delta(D,\Pi)$. 

To trace through the min-max statement in more detail, let $F \in \cF(D)$
be a feasible set. Then $\Pf A[F \cup T] \neq 0$. Recall that the
terms of the Pfaffian enumerate perfect matchings in the support graph. 
Thus $G[F \cup T]$ has a perfect matching $M$. Consider an index $i \in [k]$
such that $|T_i|$ is odd. Then either $F_i=F \cap V_i$ is odd, or an
edge of $M$ matches $F_i \cup T_i$ to $T_0$. But that can only happen
for at most $|T_0|$ indices $i$. Thus at least $\delta(D,\Pi)$
indices $i$ have $|F_i|$ odd, i.e., not a union of pairs.

\section{Algorithms for fundamental delta-matroid problems} \label{sec:enum-polys}

We now present the various algorithms over linear delta-matroids.

\subsection{Max-weight feasible sets}
\label{sec:max-weight}

We show an $O(n^\omega)$-time algorithm for finding a max-weight
feasible set in a linear delta-matroid. 

More precisely, let $D=(V,\cF)$ be a delta-matroid and
$w(v) \in \Q$, $v \in V$ a set of element weights. Let $n=|V|$.
The goal is to find a feasible set $F \in \cF$ to maximize the weight
$w(F)=\sum_{v \in F} w(v)$. Note that since the feasible sets of $D$
do not necessarily all have the same cardinality, the negative element
weights cannot easily be removed by any simple transformation (as,
e.g., shifting by a constant would affect different feasible sets differently).
This problem can be solved via the ``signed greedy'' algorithm,
which extends the normal greedy algorithm (in fact, like for matroids,
the success of signed greedy can be taken as a definition of delta-matroids)~\cite{Bouchet95,Bouchet87DMone}.
% The algorithm consists of the following steps.\todo{MW: This could easily be replaced by a literature reference.}
%\begin{enumerate}
%\item Sort the elements of $V$ as $v_1, \ldots, v_n$ such that $|w_1| \geq \ldots \geq |w_n|$,
%  where $w_i=w(v_i)$
%\item Initialize $I=J=\emptyset$
%\item Iteratively, for $i=1$ to $n$:
%  \begin{itemize}
%  \item If $w_i \geq 0$, replace $(I,J)$ by $(I+v_i,J)$ if this is separable,
%    otherwise by $(I,J+v_i)$
%  \item If $w_i < 0$, replace $(I,J)$ by $(I,J+v_i)$ if this is separable,
%    otherwise by $(I+v_i,J)$
%  \end{itemize}
%\item Return $I$ as a max-weight feasible set. 
%\end{enumerate}
However, this requires $O(n)$ calls to a separation oracle.
If $D$ is linear, this algorithm thus runs in $O(n^{\omega+1})$ time.
We show an improvement using the contraction representation.
We begin with the following observation.

\begin{lemma} \label{lemma:maxwtf-twist}
  Let $D=(V,\cF)$ be a delta-matroid and $w \colon V \to \Q$ a set of weights.
  Let $N \subseteq V$ be the set of elements $v \in V$ such that $w(v)<0$.
  Let $D'=D \Delta N$, and define a set of weights $w'$ by $w'(v)=|w(v)|$
  for every $v \in V$. For any feasible set $F \in \cF$,
  we have $w(F) = w'(F \Delta N) + w(N)$. 
\end{lemma}
\begin{proof}
  The following hold.
  \begin{align*}
    w(F) &= w'(F \setminus N)  - w'(F \cap N) \\
    w'(F \Delta N) &= w'(F \setminus N) + w'(N \setminus F)
  \end{align*}
  Thereby $w(F)-w'(F \Delta N)=-w'(N)=w(N)$ as promised. 
\end{proof}

The problem of finding a max-weight feasible set in a linear
delta-matroid in contraction representation now reduces to the
well-known problem of finding a max-weight basis of a linear matroid.

\begin{theorem} \label{thm:max-wtf}
  There is a deterministic algorithm that finds a max-weight feasible set
  in a linear delta-matroid using $O(n^\omega)$ field operations.
\end{theorem}
\begin{proof}
 Let $D=(V,\cF)$ and $w \colon V \to \Q$ be given as input, and
 as above define $N=\{v \in V \mid w(v) < 0\}$, $D'=D \Delta N$
 and $w' \colon V \to \Q$ where $w'(v) = |w(v)|$ for all $v \in V$.
 Let $\bD(A)/T$ be a contraction representation of $D'$, which
 can be constructed using Lemma~\ref{lemma:representation}. 
 Finally, order the columns of $A$ to begin with $T$ and thereafter
 elements $v \in V$ in order of non-increasing weight $w'(v)$.
 Let $B$ be a lex-min column basis for $A$ with respect to this ordering.
 Then $B$ can be computed in $O(n^\omega)$ field operations over $A$,
 and $B$ is a max-weight column basis of $A$ with respect to the weights $w'$. 
 We claim that $F=(B \setminus T) \Delta N$ is a max-weight feasible set in $D$.
 For this, let $F^*$ be a max-weight feasible set in $D$ with
 respect to the weights $w$. Then by Lemma~\ref{lemma:maxwtf-twist},
 $F^* \Delta N$ is a max-weight feasible set in $D'$ with respect to
 the weights~$w'$. Hence $B'=(F^* \Delta N) \cup T$ is feasible in $\bD(A)$
 and $w'(B \setminus T) \geq w'(B' \setminus T) = w(F^*)-w(N)$.
 On the other hand, let $B$ be a lex-min basis of $A$ in the above ordering. 
 Then $T \subseteq B$ by construction. By Lemma~\ref{lemma:maxwtf-twist},
 $F=(B \setminus T) \Delta N$ is a feasible set in $D$ with
 $w(F)=w'(B \setminus T) + w(N) \leq w(F^*)$. 
 Hence $w'(B \setminus T) = w(F^*)-w(N)$ by sandwiching
 and $w((B \setminus T) \Delta N)=w'(B \setminus T)+w(N) = w(F^*)$.
\end{proof}

\subsection{DM Parity and Delta-Covering} \label{sec:parity}

Recall that the \textsc{DM Parity} problem is defined as follows.
Let $D = (V, \cF)$ be a delta-matroid with $V$ partitioned into $n$ pairs $\Pi$.
The problem is to find a feasible set $F \in \cF$ minimizing
the number of \emph{broken} pairs
$\delta_{\Pi}(F) = |\{ P \in \Pi : |F \cap P| = 1 \}|$.
Let $\delta(D, \Pi) = \min_{F \in \cF} \delta_{\Pi}(F)$.
We consider the equivalent \textsc{DM Delta-covering}
problem defined as follows.
Let $D_1 = (V, \cF_1)$ and $D_2 = (V, \cF_2)$ be two given delta-matroids.
The problem is to find $F_1 \in \cF_1$ and $F_2 \in \cF_2$ maximizing $|F_1 \Delta F_2|$.
Let $\tau(D_1, D_2) = \max_{F_i \in \cF_i} |F_1 \Delta F_2|$.
As described in Section~1, \textsc{DM Parity} and \textsc{DM Delta-Covering} reduce to each other.
Moreover, \textsc{DM Covering} and \textsc{DM Intersection} are special cases of \textsc{DM Parity} and \textsc{DM Delta-Covering}.
% \todo{do we need this? alreadyin intor}
% We describe the equivalence between these two problems shown by Geelen, Iwata, and Murota~\cite{GeelenIM03}.
% For an instance $(D = (V, \cF), \Pi)$ of \textsc{Delta-matroid Parity}, let  $D_{\Pi} = (V, \cF_{\Pi})$ be the matching delta-matroid of the graph $(V, \Pi)$.
% For any $F \in \cF$, it holds that $\delta_{\Pi}(F) = |V| - \max_{F_{\Pi} \in \cF_{\Pi}} |F \Delta F_{\Pi}|$, since $\max_{F_{\Pi} \in \cF_{\Pi}} |F \Delta F_{\Pi}|$ counts pairs $P$ with $|F \cap P| = 1$ once and other pairs twice.
% Thus, $\delta(D, \Pi) = |V|-\tau(D, D_{\Pi})$.
% Indeed, for an optimal solution $(F, F_{\Pi})$ for \textsc{Delta-matroid Delta-cover}, we have $\delta_{\Pi}(F) = \delta(D, \Pi)$.
% Conversely, for an instance $(D_1, D_2)$ of \textsc{Delta-matroid Delta-cover}, define a delta-matroid $D$ by taking the direct sum of $D_1$ and $D_2^*$.
% Let $\Pi$ contain the pairs corresponding to copies.
% Then, $\tau(D_1, D_2) = |V| - \delta(D, \Pi)$.
% Further, if $F$ is an optimal solution for \textsc{Delta-matroid Parity}, then the pair $(V_1 \cap F, V_2 \setminus F)$ constitutes an optimal solution for \textsc{Delta-matroid delta-cover}.

We can compute $\tau(D_1, D_2)$ in $O(n^{\omega})$ field operations as follows.
Observe that $\tau(D_1, D_2)$ is the maximum feasible set size in the delta-sum $D_1 \Delta D_2$.
By Lemma~\ref{lemma:delta-sum}, we can find a linear representation $\bD(A) / T$ of $D_1 \Delta D_2$  in $O(n^{\omega})$ field operations.
We then have $\tau(D_1, D_2) = \rank A - |T|$.
Thus the decision variants of \textsc{DM Parity} and \textsc{DM Delta-Covering} can be solved in $O(n^{\omega})$ field operations.
Via self-reducibility, we obtain an algorithm that finds a witness for \textsc{DM Parity} and \textsc{DM Delta-Covering} in $O(n^{\omega + 1})$ field operations, matching the result of Geelen et al.~\cite{GeelenIM03}.
We present an improvement to $O(n^{\omega})$, using the method of Harvey \cite{Harvey09}.
% The main result in this section is the design of algorithms for \emph{search} problems can be solved in the same time bound, i.e., we give an $O(n^{\omega})$-time algorithm for finding $F_1 \in \cF_1$ and $F_2 \in \cF_2$ such that $\tau(D_1, D_2) = |F_1 \Delta F_2|$ (Theorem~\ref{thm:parity-algo}).
% For clarity, we first present an algorithm for a simpler case where we are interested in whether $\delta(D, \Pi) = 0$ or $\tau(D_1, D_2) = |V|$.

\begin{lemma}[Harvey~\cite{Harvey09}] \label{lemma:harvey}
  For a nonsingular matrix $M$, let $\widetilde{M}$ be a matrix that is identical to $M$ except that $\Delta = \widetilde{M}[X] - M[X] \ne 0$.
  \begin{enumerate}[(i)]
    \item $\widetilde{M}$ is nonsingular if and only if $I + \Delta M^{-1}[X]$ is nonsingular.
    \item If $\widetilde{M}$ is nonsingular, then $\widetilde{M}^{-1} = M^{-1} - M^{-1}[\cdot, X] (I + \Delta M^{-1}[X])^{-1} \Delta M^{-1}[X, \cdot]$.
  \end{enumerate}
  In particular, given ${M}^{-1}[X \cup Y]$, we can compute $\widetilde{M}^{-1}[Y]$ in time $\Oh(|X|^{\omega} + |X|^{\omega - 2} \cdot |Y|^{2})$ provided that $\widetilde{M}$ is nonsingular.
\end{lemma}

\begin{lemma} \label{lemma:harvey-algorithm}
  Let $A$ be an $n\times n$ skew-symmetric, non-singular polynomial matrix with its row and column indexed by $V = \{ v_1, \dots, v_n \}$.
  Suppose that $A = B + Y$, where (i) $B$ is a matrix defined over a field $\F$ containing  $\Omega(n^3)$ elements, and (ii) $Y$ is the Tutte matrix of a graph $G = (V, E)$ with variables $y_e, e \in E$.
  Suppose that $G$ has connected components $C_1, \dots, C_{\gamma}$, with $|C_i| \in O(1)$ for every $i \in [\gamma]$.
  It is possible to find an inclusion-wise maximal set $S \subseteq E$ for which that $A$ remains non-singular when setting $y_e$ to zero for all $e \in S$. This can be done with probability $1 - 1/\Omega(n)$ using $O(n^{\omega})$ field operations over $\F$.
\end{lemma}
\begin{proof}
 First, we choose $\widehat{y_e} \in \F$ uniformly at random for every $e \in E$.
 Our algorithm will work over $\F$, replacing $y_e$ with $\widehat{y_e}$.
 For the purpose of analysis, however, we will treat $A$ is a polynomial matrix.

 We first describe a simpler algorithm. 
 The idea is to set all variables $y_e$'s to zero while ensuring that $A$ remains non-singular throughout the process. 
 The algorithm iterates over each $e \in E$.
 Let $A'$ be the matrix generated from $A$ by assigning $y_{e} = 0$.
 We check whether $A'$ is non-singular in $O(n^{\omega})$ field operations.
 If $A'$ is non-singular, then update $A$ by setting $y_e$ to zero. 
 We then proceed to the next edge.
 This algorithm takes $O(n^{\omega + 1})$ field operations.
 We claim that the edges $e$ corresponding to the variables $y_e$ that are set to zero forms a desired set $S$.
 This is because $y_{e}=0$ is not set to zero if and only if $\Pf A$ is divisible by $y_e$.
 Since $\Pf A$ is multilinear in the $y$-variables throughout,
 after iterating through every $e \in E$, $\Pf A$ contains a single monomial.
 That is, $S$ is a maximal set such that $A$ remains non-singular when $y_i, i \in S$ is set to zero.
 % namely, $\sigma_{F,I} \Pf A[F \cup T] \cdot \prod_{i \in [n] \setminus I} y_i$, where $F = \bigcup_{i \in I} \{ v_{2i-1},v_{2i} \}$, and $I$ is the set of integers $i \in [n]$ with $y_i$ set to zero.
 % As we maintain the invariant that $A_i$ is non-singular, $\Pf A[F \cup T] \ne 0$, i.e., $F$ is feasible.
 
 We argue that the non-singularity $A$ can be tested more efficiently by maintaining the inverse $A^{-1}$, leading to an algorithm using $O(n^3)$ field operations.
 With $A^{-1}$ at hand, by Lemma~\ref{lemma:harvey} (i), we can check in $O(1)$ time whether $A$ remains non-singular after setting $y_{e}$ to zero.
 More precisely, for $e = v_i v_j$, $i < j$, letting $X = \{ v_i, v_j \}$ in Lemma~\ref{lemma:harvey} (i), $A$ is non-singular if and only if
 \begin{align*}
   \det \left( \begin{pmatrix} 1 & 0 \\ 0 & 1 \end{pmatrix} 
   + \begin{pmatrix} 0 & -y_{e} \\ y_{e} & 0 \end{pmatrix} \begin{pmatrix} 0 & A^{-1}[v_i, v_j] \\ -A^{-1}[v_i, v_j] & 0 \end{pmatrix} \right)
     = (1 + y_{e} A^{-1}[v_i, v_j])^2
 \end{align*}
 is nonzero, i.e., $y_e \neq -1/A^{-1}[v_{i}, v_j]$.
 We can update $A^{-1}$ using Lemma~\ref{lemma:harvey}~(ii) in $O(n^{2})$ operations because $|X| = 2$.
 Thus this algorithm uses $O(n^3)$ field operations.

 \begin{algorithm}[t]
   % \begin{minipage}{.8\linewidth}
   % \centering
   \caption{The algorithm in Lemma~\ref{lemma:harvey-algorithm}.}
   \label{algorithm:divide-conquer}
 
   \begin{algorithmic}
   \Procedure{DeleteEdges}{$I, N$}
   % \Comment{$N \in \F^{V' \times V'}$ and $V', S \subseteq V$}
     \State \textit{Invariant:} $N = A^{-1}[\bigcup_{i \in I} C_{i}]$
     \If {$|I| = \{ i \}$}
       \For{each edge $e = v_i v_j$ with $i < j$ in $G[C_i]$}
       \State \algorithmicif\ $y_{e} \ne -1 / N[v_i, v_j]$ \algorithmicthen\ $y_{e} \leftarrow 0$ and update $N[C_i]$ using Equation~\eqref{equation:update}
       \EndFor
     \Else
       \State Partition $I$ into $I'$ and $I''$ of equal size
       \State \Call{DeleteEdges}{$I'$, $N[I']$}
       \State Update $N[I'']$ using Equation~\eqref{equation:update}
       \State \Call{DeleteEdges}{$I''$, $N[I'']$}
     \EndIf
   \EndProcedure
   \end{algorithmic}
   % \end{minipage}
 \end{algorithm}

 We show how to further speed up this algorithm to $\Oh(n^{\omega})$ operations using a divide-and-conquer approach as in Harvey~\cite{Harvey09}.
 Algorithm~\ref{algorithm:divide-conquer} describes our recursive procedure \textsc{DeleteEdges}.
 The input is a pair $(I, N)$, where $I \subseteq [\gamma]$ is an interval and $N$ is a matrix indexed by $\bigcup_{i \in I} C_i$.
 Our algorithm maintains the invariant that upon invocation $N = A^{-1}[\bigcup_{i \in I}C_i]$.
 For $|I| > 1$, our algorithm partition $I$ into two sets $I'$ and $I''$ of equal size, and invoke \textsc{DeleteEdges} recursively.
 To maintain the invariant, before the second call, we update the matrix $N[I'']$  as follows using Lemma~\ref{lemma:harvey} (ii), so that $N = A^{-1}[I'']$ holds again:
 \begin{equation} \label{equation:update}
   N[I''] = N[I''] - N[I'', I'] (I + \Delta N[I''])^{-1} \Delta N[I', I''],
 \end{equation}
 where $\Delta$ is a matrix of dimension $|I'| \times |I'|$ containing all differences in the first recursion.
 Note that this update can be done in $\Oh(|I|^{\omega})$ field operations.
 For $|I| = 1$, then we go over every edge $e = v_i v_j$ in $G[C_i]$, setting $y_e$ to zero if $y_e \ne -1/N[v_{i}, v_{j}]$.
 This condition is equivalent to the non-singularity of $A$ by Lemma~\ref{lemma:harvey} (i); see above.
 If $y_e$ is set to 0, then we update $N$ according to Equation~\eqref{equation:update}.
 Let $f(n)$ be the running time of \textsc{DeleteEdges} for $n = |I|$.
 We have the recurrence $f(n) = 2f(n/2) + \Oh(n^{\omega})$, which gives $f(n) = \Oh(n^{\omega})$.
 % We can apply \textsc{DeleteEdges} for other pairs in $O(n^{\omega})$ time.

 Finally, we analyze the success probability. 
 Our algorithm fails when $A_i$ is falsely reported as singular. 
 By the Schwartz-Zippel lemma, this happens with probability $1/\Omega(n^2)$.
 By the union bound, the success probability is $1 - 1/\Omega(n)$.
\end{proof}

Using Lemma~\ref{lemma:harvey-algorithm}, we prove Theorem~\ref{ithm:target-delta-sum}, which yields algorithms for \textsc{DM Delta-Covering} and \textsc{DM Parity} as well as \textsc{DM Covering} and \textsc{DM Intersection} using $O(n^{\omega})$ field operations (Corollary~\ref{icor:search}).

\targetdeltasum*

\begin{proof}

Let us recall the construction in Lemma~\ref{lemma:delta-sum}.
The set $V^i = \{ v^i \mid v \in V \}$ for $i = 1, 2$ is a copy of~$V$, and $H$ is the graph over $V^+ = V \cup V_1 \cup V_2$ with edges $v v^1$, $v v^2$, and $v^1 v^2$ for all $v \in V$.
The Tutte matrix of $H$ is denoted by~$A_H$.
The matrix $A'$ is indexed by $V^+ \cup T_1 \cup T_2$, and we have $A'[V_1 \cup T_1] = A_1$, $A'[V_2 \cup T_2] = A_2$, and zero everywhere else.
Then $A = A_H + A'$. 
As shown in Lemma~\ref{lemma:delta-sum}, $F \subseteq V$ is feasible in $D_1 \Delta D_2$ if and only if $A[F^+]$ is non-singular, where $F^+ = F \cup V_1 \cup T_1 \cup V_2 \cup T_2$.
Let $A_S = A[F^+]$.
It follows from the proof of Lemma~\ref{lemma:delta-sum} that
\begin{align*}
  \Pf A_S =
  \sum_{\substack{F_1 \in \cF_1, F_2 \in \cF_2, \\ F_1 \Delta F_2 = F}} \sigma_{F_1,F_2} \Pf A_1[F_1 \cup T_1] \cdot \Pf A_2[F_2 \cup T_2] \cdot \prod_{v \in V \setminus (F_1 \cup F_2)} y_{v^1,v^2} \prod_{v \in F_1 \setminus F_2} y_{v,v^2} \prod_{v \in F_2 \setminus F_1} y_{v,v^1},
\end{align*}
where $\sigma_{F_1,F_2} = \pm 1$ is the sign term and the $y$-variables (i.e., $y_{v,v^1}, y_{v,v^2}, y_{v^1,v^2}$ for $v \in V$) represent the entries of $A_H$.
% The rest is analogous to the proof of Theorem~\ref{thm:parity-special}.
% We set all $y$-variables to zero, while ensuring that $A_S$ remains non-singular.
% This yields two feasible sets $F_i \in \cF_i, i = 1, 2$ such that $F_1 \Delta F_2 = F$:
By applying Lemma~\ref{lemma:harvey-algorithm}, we can find a maximal set of $y$-variables such that $A'$ remains non-singular when those are set to zero in $O(n^{\omega})$ field operations.
An element $v \in V$ belongs to (i) $F_1 \setminus F_2$ if $y_{v,v^2}$ is not set to zero, (ii) $F_2 \setminus F_1$ if $y_{v,v^1}$ is not set to zero, and (iii) $V \setminus (F_1 \cup F_2)$ if $y_{v^1, v^2}$ is not set to zero.
% As in the proof of Theorem~\ref{thm:parity-special}, by maintaining the relevant part in the inverse of $A_S$, Harvey's divide-and-conquer scheme allows us to iterate over all $y$-variables in $O(n^{\omega})$ time.
\end{proof}

\subsection{Weighted delta-matroid intersection} \label{sec:intersection}

Let $D_1 = (V, \cF_1)$ and $D_2 = (V, \cF_2)$ be two linear delta-matroids with weights $w(v) \in \N, v \in V$.
We consider the \emph{intersection} problem, where the goal is to find a common feasible set $F \subseteq V$, i.e., $F \in \cF_1$ and $F \in \cF_2$.
The decision problem, i.e., whether there exists $F \in \cF_1 \cap \cF_2$, can be solved in $\Oh(n^{\omega})$ by testing whether $V$ is feasible in the delta-sum $D_1 \Delta D_2^*$.
Moreover, we can find a common feasible set in $O(n^{\omega})$ time using the algorithm in Section~\ref{sec:parity}.
In this section, we give a (pseudo)polynomial-time algorithm for the \textsc{Weighted Delta-matroid Intersection}, where we are tasked with finding a common feasible set of maximum weight, answering an open question of Kakimura and Takamatsu~\cite{KakimuraT14sidma}.
To the best of our knowledge, there has been no polynomial-time algorithm even for the unweighted case.

\cardinalityintersection*

\begin{proof}

Suppose that the contraction representation $D_i = \bD(A_i) / T_i$ is given for $i = 1, 2$.
Let $V_i = \{ v_i \mid v \in V \}$ for $i = 1, 2$ be two copies of~$V$, and define matrices 
\begin{align*}
  A' = \kbordermatrix{
    & V_1 & T_1 & V_2 & T_2 \\
    V_1 & A_1[V_1] & A_1[V_1, V_1] & O & O \\
    T_1 & A_1[T_1, V_1] & A_1[T_1] & O & O \\
    V_2 & O & O & A_2[V] & A_2[V, T_2] \\
    T_2 & O & O & A_2[T_2, V] & A_2[T]
  } 
  \text{ and }
  A_H = \kbordermatrix{
    & V_1 & T_1 & V_2 & T_2 \\
    V_1 & O & O & B & O \\
    T_1 & O & O & O & O \\
    V_2 & -B & O & O & O \\
    T_2 & O & O & O & O 
  }
\end{align*}
where $B$ is the diagonal matrix where the entry corresponding to $v$ is $z^{w(v)} y_v$ for indeterminates $y_v$ and~$z$.
Let $A = A' + A_H$.

To compute the determinant of $A$, we will use the algorithm of 
Storjohann \cite{Storjohann03}.
For a matrix $A \in \F[z]^{n \times n}$ whose entries are polynomials of degree at most $d$, the determinant $\det A$ can be computed in $\Oh(d n^{\omega})$ arithmetic operations with high probability.
% We will also use the fact that the product of two polynomials of degree at most $d$ can be computed in $d^{1 + o(1)}$ arithmetic operations (see e.g., \cite{CantorK91,von2013modern}).

To determine the maximum common feasible set weight,
our algorithm first constructs the polynomial $\det A$ using the algorithm of Storjohann.
Let us analyze $\Pf A$ using Lemma~\ref{lemma:sum-pf}.
For $V' \subseteq V$ and $i = 1, 2$, let $V_i' = \{ v_i \mid v \in V \}$ be the set containing copies of $v \in V$ in $V_i$, and let $w(V') = \sum_{v \in V'} w(v)$.
For $X \subseteq V_1 \cup T_1 \cup V_2 \cup T_2$, note that $A_H[X]$ is non-singular if and only if $X$ consists of pairs $\{ v_1, v_2 \}$, i.e., $X = V_1' \cup V_2'$ for some $V' \subseteq V$, in which case $\Pf A_H[V_1' \cup V_2'] = z^{w(V')} \prod_{v \in V'} y_v$.
Thus, we have
\begin{align}
  \Pf A
  &= \sum_{V' \subseteq V} \sigma_{V'} \Pf A'[V_1' \cup T_1 \cup V_2' \cup T_2] \Pf A_H[(V_1 \setminus V_1') \cup (V_2 \setminus V_2')] \nonumber \\
  &= \sum_{V' \subseteq V} \sigma_{V'} \Pf A_1[V_1' \cup T_1] \Pf A_2[V_2' \cup T_2] \cdot z^{w(V \setminus V')} \prod_{v \in V \setminus V'} y_v,  \label{eq:intersection}
\end{align}
where $\sigma_{V'} = \pm 1$ denotes the sign.
If there is a set $F \in \cF_1 \cap \cF_2$ of weight $t$, then the coefficient of $z^{w(V) - t}$ in $\Pf A$ is non-zero (i.e., its order is $w(V) - t$).
Since $\det A = (\Pf A)^2$, the maximum weight of feasible sets is the largest integer $t$ such that the coefficient of $z^{2w(V) - 2t}$ in $\det A$ is non-zero.
We can thus compute the maximum common feasible set weight in $\Oh(Wn^{\omega})$ time.
\end{proof}

\paragraph*{Ishikawa-Wakayama formula.}
Finally, we show that the Ishikawa-Wakayama (Lemma~\ref{lemma:cauchy-binet-ss}) formula follows from the
proof of Theorem~\ref{ithm:cardinality-intersection}.

We show that the proof of Theorem~\ref{ithm:cardinality-intersection} gives rise to the Ishikawa-Wakayama formula when applied to the intersection between a directly represented delta-matroid and a matroid.
Consider a directly represented delta-matroid $D = \bD(A)$ and a matroid represented by $B$ ($\bD(B') / T$ in the contraction representation).
As in the proof of Theorem~\ref{ithm:cardinality-intersection}, construct a matrix
\begin{align*}
 C = \kbordermatrix{ & V_1 & V_2 & T \\ V_1 & A & I & O \\ V_2 & -I & O & B^T \\  T & O & -B & O },
\end{align*}
where $V_1 = \{ 1, \dots, 2n \}$, $V_2 = \{ 2n + 1, \dots, 4n \}$, and $T = \{ 4n + 1, \dots, 4n + 2k \}$.
Note that indeterminates $y_v, v \in V$ and $z$ are all set to 1. 
We claim that 
\begin{align*}
\Pf C = (-1)^n \sum_{S \in \binom{[2n]}{2k}} \det B[\cdot, V'] \cdot \mathrm{Pf} \, A[V'] \text{ and } \Pf C = (-1)^n \Pf BAB^T. 
\end{align*}
For the first claim, by Lemma~\ref{lemma:sum-pf}, we have
\begin{align*}
 \Pf C = \sum_{V' \in \binom{[2n]}{2k}} \sigma_{V'} \Pf C[(V_1 \setminus V_1') \cup (V_2 \setminus V_2')] \cdot \Pf C[V_1'] \cdot \Pf C[V_2' \cup T],
\end{align*}
where $V_1' = V'$, $V_2' = \{ i + 2n \mid i \in V' \}$, and $\sigma_{V'}$ is the sign of the permutation that brings $(V_1 \setminus V_1', V_2 \setminus V_2', V_1', V_2', T)$ in ascending order.
Notably, this sign is equal to the sign of the permutation that arranges $(V_1 \setminus V_1', V_1', V_2 \setminus V_2', V_2', T)$ in ascending order as it can be obtained by exactly $2k(2n - 2k)$ transpositions, and thus $\sigma_{V'} = 1$.
Moreover, by Lemma~\ref{lemma:det-pf}, $\Pf C[V_1' \cup V_2'] = (-1)^{\frac{1}{2}(2n - 2k)(2n - 2k - 1)} = (-1)^{n+k}$ and $\Pf C[V_2' \cup T] = (-1)^{\frac{1}{2}(2k)(2k-1)} \cdot \det B[V'] = (-1)^k \cdot \det B[V']$.
Thus, the first claim follows.
Now we proceed to the second claim.
Note that by Lemma~\ref{lemma:det-pf},
\begin{align*}
 \Pf \begin{pmatrix}
   A & I \\ -I & O
 \end{pmatrix} = (-1)^{\frac{1}{2} (2n)(2n - 1)} = (-1)^n
 \text{ and } \begin{pmatrix}
   A & I \\ -I & O
 \end{pmatrix}^{-1} = \begin{pmatrix} O & I \\ I &  A \end{pmatrix}
\end{align*}
We use the following fact about the Schur complement:
\begin{lemma}[Schur complement] \label{lemma:schur} Consider a skew-symmetric matrix $A$ of the form
  \[
  A =
  \begin{pmatrix}
    B & C \\
    -C^T & D
  \end{pmatrix},
  \]
  where $B$ is a non-singular square matrix.
  Then, $\Pf A = \Pf B \cdot \Pf \, (D + C^T B^{-1} C)$.
\end{lemma}
By Lemma~\ref{lemma:schur},
\begin{align*}
 \mathrm{Pf}\, C = \mathrm{Pf} \begin{pmatrix}
   A & I \\ -I & O
 \end{pmatrix} \cdot \mathrm{Pf}\, \left( \begin{pmatrix} O & B \end{pmatrix}  \begin{pmatrix}
   A & I \\ -I & O
 \end{pmatrix} \begin{pmatrix} O \\ B^T \end{pmatrix} \right) = (-1)^n \Pf BAB^T.
\end{align*}
Thus, Ishikawa-Wakayama formula has been proven.

\section{Conclusions}

We have shown a new representation for linear delta-matroids, the
\emph{contraction} representation, and used it to derive a range of
new results, including linear representations of delta-matroid union
and delta-sum, and faster algorithms for a range of problems including
\textsc{Linear Delta-Matroid Parity}. We also show the first polynomial-time
algorithms for the maximum cardinality and weighted versions of
\textsc{Linear Delta-Matroid Intersection} (in the latter case with a
pseudopolynomial running time), solving an open question of
Kakimura and Takamatsu~\cite{KakimuraT14sidma}.

We note a few open questions. First, all our running times are stated
purely in terms of the number of elements $n$. It would be interesting
to investigate faster algorithms for linear delta-matroids under some
notion of bounded rank, like for linear matroids. But we also note two
specific challenging questions.
\begin{enumerate}
\item Is there a strongly polynomial-time algorithm for
  \textsc{Weighted Linear Delta-Matroid Parity}? This would extend the recent
  result for \textsc{Weighted Linear Matroid Parity}~\cite{IwataK22SICOMP}.
\item Is there a $\tilde O(Wn^\omega)$-time algorithm for \textsc{Shortest
    Disjoint $\cS$-Paths}? This would extend results for graph matching~\cite{CyganGS15matching}.
\end{enumerate}
Additionally, does there exist a good characterization of the maximum
cardinality version of \textsc{Linear Delta-Matroid Intersection}?

\paragraph*{Acknowledgments.} TK is supported by the DFG project DiPa (NI 369/21). TK would like to thank Taihei Oki for enlightening discussions on delta-matroids. 

\bibliographystyle{abbrv}
\bibliography{all}

\newpage
\appendix

\end{document}